\documentclass[12pt]{article}
\usepackage{amsmath,amsfonts, amssymb, color,ulem}
\usepackage{pdfsync}

\numberwithin{equation}{section}

\newtheorem{theorem}{Theorem}[section]
\newtheorem{definition}[theorem]{Definition}
\newtheorem{proposition}[theorem]{Proposition}
\newtheorem{corollary}[theorem]{Corollary}
\newtheorem{lemma}[theorem]{Lemma}

\newtheorem{assump}[theorem]{Assumption}

\newtheorem{coro}{\sc Corollary}[section]

\newenvironment{proof}{\addvspace{\medskipamount}\par\noindent{\it Proof}.}
{\unskip\nobreak\hfill$\Box$\par\addvspace{\medskipamount}}

\newcommand{\p}{\mathbb{P}}
\newcommand{\Q}{\mathbb{Q}}
\newcommand{\cE}{{\mathcal E}}
\newcommand{\E}{\mathbb{E}}
\newcommand{\dpm}{{\prime\hspace{-0.03cm}\prime}}

\newcommand{\cF}{{\mathcal F}}

\newcommand{\R}{\mathbb{R}}
\newcommand{\id}{{\mathbf 1}}

\newcommand{\da}[2]{\frac{{\rm d} #1}{{\rm d} #2}}

\newcommand{\pf}{{\it Proof: }}
\newcommand{\eof}{\hfill {\it Q.E.D.} \vspace*{0.3cm}}

\begin{document}
  \title{{Consistent Investment of Sophisticated Rank-Dependent Utility Agents
  in Continuous Time
   }}


\author{Ying Hu\thanks{Univ Rennes, CNRS, IRMAR-UMR 6625, F-35000 Rennes, France; Email: ying.hu@univ-rennes1.fr. Partially supported by Lebesgue Center of Mathematics “Investissementsd’avenir”program-ANR-11-LABX-0020-01, by ANR CAESARS (Grant No. 15-CE05-0024) and by ANR MFG (GrantNo. 16-CE40-0015-01).} \and
Hanqing Jin\thanks{Mathematical Institute and Oxford--Nie Financial Big Data Lab, 
University of  Oxford, Oxford OX2 6GG, UK; Email: jinh@maths.ox.ac.uk. The research of this author was  partially supported by  research
grants from  Oxford--Nie Financial Big Data Lab and  Oxford--Man Institute of Quantitative Finance. }
\and
Xun Yu Zhou\thanks{Department of IEOR, Columbia University, New York, NY 10027, USA;  Email: xz2574@columbia.edu. The research of this author was supported through start-up grants at both University of Oxford and Columbia University, and research funds from Oxford--Nie Financial Big Data Lab,
Oxford-Man Institute of Quantitative Finance, and Nie Center for Intelligent Asset Management. }}

\maketitle

\begin{abstract}
We study portfolio selection in a complete continuous-time market where the preference is dictated by the
rank-dependent utility. As such a model is inherently time inconsistent due to the underlying
probability weighting, we study the investment behavior of
sophisticated consistent planners who seek (subgame perfect) intra-personal equilibrium strategies. We provide sufficient conditions
under which an equilibrium strategy is a replicating portfolio of
a final wealth. We derive this final wealth profile explicitly,
which turns out to be
in the same form as in the classical
Merton model with the market price of risk process properly scaled by a deterministic function in time.
We present this scaling function explicitly through the solution to  a highly nonlinear and singular ordinary
differential equation, whose existence of solutions is established. Finally, we
give a necessary and sufficient condition for the scaling function to be smaller than 1 corresponding to an effective reduction in risk premium due to probability weighting.

\smallskip

\end{abstract}

\section{Introduction}

The classical expected utility theory (EUT) is unable to explain many puzzling phenomena  and paradoxes, such as, just to name a few,  the Allais paradox, the co-existence of  risk-averse and risk-seeking
behavior of a same individual, and the disposition effect in financial investment.
Rank-dependent utility theory (RDUT), first proposed by Quiggin (1982), has been developed to address some of these puzzles and has thus far been widely considered as
one of the most prominent alternative theories on preferences and choices. In addition to  a concave outcome utility function as in EUT, RDUT features  a probability weighting (or distortion) function whose slopes give uneven weights on
random outcomes when calculating the mean. The presence of probability weighting is supported by a large amount of experimental and empirical studies, many of which find that such a weighting function typically displays an ``inverse S-shape" (namely, it is first concave and then convex in its domain); see, e.g. Tversky and  Kahneman (1992), Wu and Gonzalez (1996) and Tanaka et al (2010). This particular shape  captures individuals' tendency to exaggerate the tiny  probabilities of both
improbable large gains (such as winning a lottery) and  improbable large losses (such as encountering a plane crash). In particular, Tversky and  Kahneman (1992) propose a specific parametric
class of inverse S-shaped weighting functions, which will be used as a baseline example to test the assumptions in our paper.

In this paper we study a continuous-time financial portfolio selection model in which an agent pursues the highest rank-dependent utility (RDU) value. In contrast to  classical continuous-time portfolio models such as Merton's (Merton 1969), a dynamic RDU model is intrinsically {\it time inconsistent}; namely, any ``optimal" strategy for today will generally not be optimal for tomorrow. As a result, there is no notion of a {\it dynamically optimal strategy} for a time-inconsistent model because any such a strategy, once devised for this moment,  will {\it have to} be abandoned immediately (and  indeed infinitesimally)  at the next moment. The time inconsistency of the RDU model emanates from the probability weighting which weights the random outcomes {\it unevenly} according to their probabilities of occurrence. Consider as an example a 10-period binomial lattice model with equal probabilities of moving up and down at any given state. Standing at $t=0$ the probability of reaching the top most state (TMS) at $t=10$ is extremely small ($2^{-10}$). If the agent has an inverse S-shaped probability weighting then she will greatly inflate this probability. As time goes by and the agent moves along the lattice the probabilities of this same event -- eventually reaching the TMS -- keep changing, and so do the
{\it degrees} of probability weighting. Indeed,
once at $t=9$ the probability of finally reaching the TMS  is either 1/2 or 0 -- which any reasonably intelligent individual is able to tell -- and hence there is no probability weighting at all. We hereby see an inconsistency in the {\it strength}  of probability weighting over time which is the key reason behind the time inconsistency of an RDU model.

Time inconsistency changes fundamentally the way we deal with dynamic optimization in general. Optimization is intimately associated with decision making, and finding an optimal solution is therefore  to advise on the best decisions. Now, there is {\it no} optimal solutions under  time inconsistency: so what is the purpose of studying a time-inconsistent problem?

In his seminal paper, Strotz (1955) describes three types of agents when facing time inconsistency. Type 1, a ``spendthrift" (or a naivet\'e as in the more recent literature), does not recognize the time inconsistency and at any given time seeks an optimal solution for that moment only. As a result, his strategies are  always myopic and change all the times. The next two types are aware of time inconsistency
but act differently. Type 2 is a ``precommitter" who solves the optimization problem only {\it once} at time 0 and then commits to the resulting strategy throughout, even though she knows that the original solution may no longer be optimal at later times. Type 3 is a ``thrift" (or a sophisticated agent) who is unable to precommit and realizes that her future selves will disobey whatever plans she makes now. Her resolution is to  compromise and choose {\it consistent planning} in the sense that she optimizes taking the future disobedience as a {\it constraint}.
The Strotzian approach to time inconsistency is, therefore, {\it descriptive}, namely to describe what people actually do -- the different reactions and behaviors in front of  the inconsistency, as opposed to being {\it normative}, namely to tell people what to do.

It is both interesting and challenging to formulate mathematical models for {\it each} of the three types and solve them. It is interesting because these models are very different from the classical stochastic control based ones and different from each other, and it is challenging because the most powerful tools for tackling dynamic optimization such as dynamic programming and martingale analysis are based on time consistency and hence fail for time-inconsistent problems.
Recently, there is an upsurge of research interest and effort in the fields of stochastic control and mathematical finance/insurance in studying time-inconsistent models, mostly focusing  on three different problems: mean--variance portfolio selection, and those involving non-exponential discounting or probability weighting.
Earlier works focused on Type 2, precommitted agents (see, e.g., Li and Ng 2000, Zhou and Li 2000, He and Zhou 2011), and later ones gradually shifted to Type 3, consistent planners or sophisticated agents (Ekeland and Lazrak 2006, Bjork and Murgoci 2010,
Hu et al 2012, Hu et al 2017, Bjork et al 2014).

The Type 3 problem can be mathematically formulated as a game in the following way.
The sophisticated agent,  anticipating the disagreement between her current and
future selves,  searches for a dynamic strategy that all the future selves have no incentive to deviate from.
The resulting strategy is a (subgame perfect) {\it intra-personal equilibrium} that will be carried through. In the continuous-time setting, Ekeland and Lazrak (2006) are the first to give a formal definition of such an equilibrium (albeit for a deterministic
Ramsey model with non-exponential discounting),  based on a first-order condition of a ``spike variation" of the equilibrium strategy.

This paper derives  the consistent investment strategies of a Type 3 RDU agent in a continuous-time market in which the asset prices are described by stochastic differential equations (SDEs).
We make several contributions, both methodologically and economically.
First of all, to our best knowledge, this paper is the first to formulate and attack
an RDU investment problem  in the continuous-time setting. There are substantial difficulties in approaching the problem. Most notably, in deriving an equilibrium strategy, one needs to analyze the effect of the aforementioned spike variation on the RDU objective functional. In the absence of probability weighting, there is a well-developed approach at disposal to do this, based on calculus of variations and {\it sample-path} dependence of SDEs on parameters.\footnote{The same approach is used to derive the stochastic maximum principle for optimal stochastic controls; see, e.g., Yong and Zhou (1999).}  With probability weighting, however, we need to study the {\it distributional} dependence of the SDE solutions on parameters, which is actually a largely unexplored topic to our best knowledge. Starting from scratch, we carry out a delicate analysis to solve the problem thoroughly. The final solution requires  the existence of solutions to a highly nonlinear, singular ordinary differential equation (ODE),  the validity of an inequality, and the existence of the Lagrange multiplier to a budget constraint.  We then provide sufficient conditions on the model primitives to ensure that these three requirements are met.

The equilibrium strategy is a replicating portfolio to a terminal wealth profile we derive explicitly, assuming that the market is complete. Curiously, the terminal wealth is of the same form as that with the classical Merton optimal strategy under EUT, except that the market price of risk or risk premium process needs to be multiplied by a scaling factor, the latter in turn determined by a solution to the aforementioned nonlinear ODE. To put this in a different way, the sophisticated RDU behaves as if she was an EUT agent, albeit in a {\it revised} market where the risk premium is properly scaled.
The scaling factor depends on the original investment opportunity set  and the agent's probability weighting function, but {\it not} her outcome utility function. This suggests that the additional constraint arising from the consistent planning due to time inconsistency can be transferred to
the market opportunity set. This observation could be a key leading to identifying  market equilibria where all the agents are RDU consistent planners. 

When the scaling factor is less than 1, the risk premium is reduced and the agent 
is more risk averse than her EUT counterpart. In this case the probability weighting and consistent planning make the agent to take less risky exposure. We present a necessary and sufficient condition for this to happen.

It should be noted that, in order to derive our main results,  we make several assumptions on the model primitives. Some of them may look quite technical especially those on the weighting function. However, we do not impose those assumptions for mathematical convenience; instead we make sure that they are mathematically mild and economically reasonable. In particular, all the assumptions on the weighting function are satisfied by the baseline function proposed by Tversky and  Kahneman (1992).

%
%
%


The rest of the paper is organized as follows. In Section 2, we state our problem and define the intra-personal equilibrium
 in the same spirit as in our previous work Hu et al. (2012,2017). In Section 3, we present sufficient conditions under which the equilibrium terminal wealth can be explicitly derived. In Section 4, we examine these sufficient conditions closely on our model primitives. Section 5 is devoted to a concrete example demonstrating the results of the paper. In Section 6, we present an equivalent condition for 
 an effectively reduced risk premium. Finally, Section 7 concludes. In Appendices, we state related general results on a class of singular ODEs, and verify our assumptions on a family of time-varying  Tversky-- Kahneman's weighting functions.


\section{Problem Formulation}

\subsection{The market}
We consider a continuous-time market in a finite time horizon $[0, T]$,
where there are a  risk-free asset and $n$ risky assets being traded frictionlessly with price processes
$S_{0}(\cdot)$ and $S_{i}(\cdot)$, $i=1,\cdots, n$, respectively.
The risk-free interest rate, without loss of generality, is set to be $r(\cdot)\equiv 0$,
or equivalently,  $S_{0}(t)\equiv S_{0}(0)$.
The dynamics of $S_{i}(\cdot)$, $i=1,\cdots, n$,  of the risky assets follow  a multi-dimensional geometric Brownian motion
$$dS_{i}(t)=S_{i}(t)\left[\mu_{i}(t)dt+\sum_{j=1}^{n}\sigma_{i,j}(t)dW_{j}(t)\right],\; i=1,\cdots, n,$$
where $\mu_i(\cdot)$ and  $\sigma_{i,j}(\cdot)$ are all 
deterministic functions of time $t$, $W(\cdot)=(W_{1}(\cdot),\cdots, W_{n}(\cdot))^{\top}$ is an $n$-dimensional standard Brownian motion in a filtered probability space
$(\Omega, \cF, (\cF_{t})_{t\in [0, T]}, \p)$ with $\cF_{t}=\sigma(W_{s}: s\le t)\vee {\mathcal N}(\p)$. Here $^\top$ denotes the
matrix transpose.

Denote $\mu(t)=(\mu_{1}(t), \cdots, \mu_{n}(t))^{\top}$ and $\sigma(t)=(\sigma_{i,j}(t))_{n\times n}$.
We assume that $\sigma(\cdot)$ is invertible and that the market price of risk is
$\theta(t):=\sigma(t)^{-1}\mu(t)$; so the market is arbitrage free and complete. 

A trading strategy is a self-financing portfolio described by an $\cF_{t}$-adapted process
$\pi(\cdot)=(\pi_{1}(\cdot),\cdots, \pi_{n}(\cdot))^{\top}$, where $\pi_i(t)$ is the dollar amount allocated to
asset $i$ at time $t$. Under such a portfolio $\pi(\cdot)$, the dynamics of the corresponding wealth process $X(\cdot)$ evolve according to
the wealth equation
\begin{equation}\label{wealth}
dX(t)=\pi(t)^{\top}\mu(t)dt+\pi(t)^{\top}\sigma(t)dW(t).
\end{equation}

For any $t\in[0,T)$, we say a trading strategy $\pi(\cdot)$ is admissible on $[t,T]$ if $\E\int_t^T|\pi(s)^{\top}\mu(s)|ds<+\infty$ and $\E\int_{t}^{T}|\pi(s)^{\top}\sigma(s)|^2ds<+\infty$. 



\subsection{Rank-dependent utility}
An agent,  with an initial endowment $x_{0}$ at time $t=0$, pursues the  highest
possible  rank-dependent utility (RDU) of her wealth at $t=T$ by dynamically trading in the market.
At any given time $t$ with the wealth state $X(t)=x$ she takes an admissible portfolio $\pi(\cdot)$ on $[t,T]$ leading to the terminal wealth $X({T})$.
The RDU value of this terminal wealth is
\begin{eqnarray*}
J(X(T); t, x)&=&\int_0^{+\infty}w\left(t, \p_t(u( X(T))>y)\right)dy+\int_{-\infty}^0\left[w\left(t, \p_t(u( X(T))>y)\right)-1\right]dy,
\end{eqnarray*}
where $w(t, \cdot)$ is the probability weighting applied at time $t$, $u( \cdot)$ is the (outcome) utility function,
and  $\p_t:=\p(\cdot|\cF_{t})$
 denotes the conditional probability given $\cF_t$, which includes the
  information $X(t)=x$. 
The agent's original objective is to maximize $J(X(T); t, x)$ by choosing a proper admissible investment strategy.

Here we allow  the weighting function $w$ to depend explicitly on time $t$. This  is not just for mathematical generality; the time variation of probability weighting is supported empirically (Dierkes 2013, Cui et al 2020) and argued for on a psychology ground (Cui et al 2020).

\subsection{Equilibrium strategies}
As discussed in Introduction,  probability weighting in general causes time-inconsistency. This means
an optimal strategy with respect to $J(X(T); 0, X(0))$ is not necessarily still optimal
with respect to $J(X(T); t, X(t))$ for a future time $t>0$, where $X(\cdot)$ is the ``optimal" wealth process projected at $t=0$ for the objective $J(X(T); 0, X(0))$.
In this paper, we study the behaviors of a sophisticated  agent who is aware of the time-inconsistency  but lacks commitment, and  who instead seeks a consistent investment among all the different selves $t\in[0,T]$ by finding intra-personal equilibrium strategies.

Precisely,
given an admissible trading strategy $\pi(\cdot)$ with the corresponding wealth process $X(\cdot)$ starting from $X({0})=x_{0}$, a time $t\in[0,T)$ and a small number $\varepsilon\in(0,T-t)$,
we define
a slightly perturbed strategy which adds $\$k$  on top of $\pi(\cdot)$ over the small time interval $[t, t+\varepsilon)$  while keeping $\pi(\cdot)$  unchanged
outside of this interval. Here $k$ is  an $\cF_t$-measurable random {\it vector}.
This technique is called a {\it spike variation}. 
It follows from the wealth equation (\ref{wealth}) that the perturbed terminal wealth
is $X(T)+k^\top\Delta(t,\varepsilon)$, where $\Delta(t,\varepsilon)=
\int_t^{t+\varepsilon}\mu(s)ds+\int_t^{t+\epsilon}\sigma(s)dW(s)$.

\begin{definition}\label{equdef}
An admissible strategy $\pi(\cdot)$ with the wealth process $X(\cdot)$ starting from $X({0})=x_{0}$
is called an equilibrium strategy, if for any $t\in[0,T)$ and any $\cF_t$-measurable random vector
$k$, one has
\begin{equation}\label{key}
\limsup_{\varepsilon\downarrow 0} \frac{J(X(T)+k^\top\Delta(t,\varepsilon);t,X(t))-J(X(T);t,X(t))}{\varepsilon}\le 0.
\end{equation}
\end{definition}
This definition follows most existing works on time-inconsistent optimal controls in continuous time; see e.g. Ekeland and Lazrak (2006), Bjork and Murgoci (2010), and Hu et al (2012, 2017).
Theoretically, an equilibrium strategy is an ``infinitesimally" sub-game perfect equilibrium among
all the selves  $t\in [0,T)$.\footnote{Definition \ref{equdef} is also in line with the original Strotz's vision of intrapersonal
equilibrium (Strotz 1955), adjusted for the continuous-time setting.  More specifically, a Strotzian equilibrium strategy stipulates that for any given self $t$, all the
future selves will commit to the strategy, because any deviation from the strategy will make the deviating self
$t$ worse off. In a continuous-time setting, however, any fixed $t$ {\it alone} has no influence on the terminal wealth because it has a measure of zero. Therefore,
 one considers instead a small ``alliance" of self $t$: the
interval $[t, t+\varepsilon)$. Definition \ref{equdef}  posits that a ``deviation-in-alliance" from equilibrium
fares worse in a first-order sense.}

The objectives of this paper are to find conditions under which there exists an equilibrium strategy,  to derive
the terminal wealth under such a strategy, and to draw economic implications and interpretations of the results. Since the market is complete, once the desired terminal wealth
profile is identified its replicating portfolio is then the corresponding equilibrium strategy for consistent investment.

\subsection{Assumptions}

In this subsection we collect all the assumptions needed on the market parameters as well as on the preference functions.\footnote{Throughout this paper,
by an ``increasing'' function we mean a ``non-decreasing'' function, namely $f$ is increasing if $f(x)\geq f(y)$ whenever $x>y$.
We say $f$ is ``strictly increasing'' if  $f(x)> f(y)$ whenever $x>y$. Similar conventions are used for ``decreasing'' and ``strictly decreasing'' functions.}
These assumptions are henceforth in force, without necessarily being mentioned again in all the subsequent statements of results.

\begin{assump}\label{rcont-mkt}
The functions $\mu(\cdot)$, $\sigma(\cdot)$ and $\theta({\cdot})$  are all right continuous and uniformly bounded, and $\theta(t)\neq0$ $\forall t\in [0,T]$.
\end{assump}

\begin{assump}\label{basicassump}
\begin{itemize}
\item [{\rm (i)}] $w(\cdot, \cdot): [0,T]\times [0, 1]\mapsto [0,1]$ is measurable. Moreover, for each $t\in[0,T]$,
$w(t, \cdot): [0, 1]\mapsto [0,1]$ is strictly increasing, $C^{2}$, and $w(t, 0)=0, w(t,1)=1$ $\forall t\in [0, T]$. 
\item [{\rm (ii)}] $u(\cdot): \R\mapsto \R$ is strictly increasing, strictly concave, $C^{3}$, and $u( -\infty)=-\infty$,
$u'( -\infty)=+\infty, u'(+\infty)=0$.
 \end{itemize}
\end{assump}

These two are very weak assumptions, representing some ``minimum requirements" for the
model primitives.

Define $I(x):=(u')^{-1}(x)$ , $x>0$, and  $l(x):=-\ln u'(x)$, $x\in\R$.
\begin{assump}\label{uandw}
\begin{itemize}
\item [{\rm (i)}] 
There exists $\alpha>0$ such that  $\limsup_{x\rightarrow +\infty}\frac{x^{-\alpha}}{-I'(x)}<+\infty$,
or equivalently, $\limsup_{x\rightarrow -\infty}\frac{-u^\dpm(x)}{u'(x)^\alpha}<+\infty$.
\item [{\rm (ii)}] There exist  $a>0, b>0$ such that
$l'(x)\le e^{a|l(x)|+b}, l^{\dpm}(x)\le e^{a|l(x)|+b}$ $\forall x\in \R$.
\item [{\rm (iii)}] $u'''(x)\geq0$ $\forall x\in\R$.
\item [{\rm (iv)}] For any $t\in [0, T]$, there exist $c>0$ and $m\in (-1, 0)$, both possibly depending on $t$, such that
$$ w_p'(t,p)\le c[p^{m}+(1-p)^{m}] \quad \forall \, p\in (0,1).$$
\end{itemize}
\end{assump}

Assumption \ref{uandw}-(i) and -(ii) are very mild conditions satisfied by most commonly used utility functions defined on $\R$ (e.g. the exponential utility).
Assumption \ref{uandw}-(iii) implies the {\it risk prudence} of the agent, capturing her tendency to take precautions against future risk. Many common 
utility functions are prudent.\footnote{Through an experiment with a large number of subjects, Noussair et al (2014)
observe that the majority of individuals' decisions are consistent with prudence.}
Assumption \ref{uandw}-(iv) is to control the level of probability weighting on
very small and very large probabilities. It is satisfied by some well-known weighting functions, e.g. that of Tversky  and Kahneman (1992); see Appendix \ref{KT-check}.

\begin{assump}\label{swappingextra}
There exist constants  $\nu>0$, $\zeta>0$, such that
$$\int_{-\infty}^{0}u'(x)^{-\nu}dx<+\infty, \quad \int_0^{\infty}u'(x)^{\nu}dx<+\infty,\;\; \int_{-\infty}^{+\infty}u'(x)e^{-\zeta x^{2}}dx<+\infty.$$
\end{assump}

This is also a mild assumption, which is satisfied by, say, the exponential utility function.

\subsection{A crucial function}

The following real-valued function, generated from the weighting function,  will play a central role throughout  this paper:
$$h(t,x):=\E\left[w_p'(t, N(\xi))e^{x\xi}\right]>0,\;\; t\in[0,T],\; x\in\R,$$
where (and henceforth) $w_p'(t,p):= \frac{\partial}{\partial p}w(t,p)$, $\xi$ is a standard normal random variable with $N$ being its probability distribution function.
Similarly, we denote by $h'_x(t,x)$ and $h''_x(t,x)$ the first- and second-order partial derivatives of $h$ in $x$ respectively, and
in general $h^{(n)}_x(t,x)$ the $n$-th order partial derivative of $h$ in $x$.



\begin{lemma}\label{2.5to3.1.ii}
 $h(t,x)<+\infty$ $\forall (t,x)\in [0,T]\times \R.$
\end{lemma}
\begin{proof}
Fix $(t,x)\in [0,T]\times \R.$ It follows from Assumption \ref{uandw}-(iv) that
 $w_p'(t,p)\le c(p^{m}+(1-p)^{m}) \quad \forall \, p\in (0,1)$ for some $c>0$ and $m\in (-1, 0)$.
Denote $\hat m:=\frac{m-1}{2}\in (-1,m)$,  $p:=\frac{\hat m}{m}>1$,  $q:=\frac{p}{p-1}>1$.
 Applying Young's inequality $a^m b\le \frac{a^{\hat m}}{p}+\frac{b^q}{q}$  $\forall a\ge0,\;b\ge0$, we obtain
$$
 \E[(N(\xi))^m e^{\pm x\xi}]\le \E\left[\frac{1}{p} (N(\xi))^{\hat m}+\frac{1}{q} e^{\pm xq\xi}\right]
= \frac{1}{(\hat m+1) p}+\frac{1}{q} e^{(qx)^2/2}.$$
So
\begin{eqnarray*}
h(t,x)&\le& c\left(\E[N(\xi)^me^{x\xi}]+\E[ N(-\xi)^m e^{x\xi}]\right)\\
&=& c(\E[N(\xi)^me^{x\xi}]+\E[ N(\xi)^m e^{-x\xi}])\\
&\le&2c\left[ \frac{1}{(\hat m+1) p}+\frac{1}{q} e^{(qx)^2/2}\right]<+\infty.
\end{eqnarray*}
The proof is complete.
\end{proof}

The following lemma collects some basic properties of $h$ which will be useful  in the sequel.
\begin{lemma}\label{hregularity} 
For any $t\in [0, T]$, $h(t,\cdot)$ has the following properties:
\begin{enumerate}
 \item[{\rm (i)}] $h(t,0)=1$.
 \item[{\rm (ii)}] $h(t,x)$ is $C^{\infty}$ in $x\geq0$, with
  $h^{(n)}_x(t,x)=\E\left[w_p'(t, N(\xi))\xi^{n}e^{x\xi}\right],\;x\geq0$, $\forall n\geq0$.
  \item[{\rm (iii)}] $h^{(n)}_x(t,x)$ is convex in $x\geq0$ for all even $n\geq0$, and increasing in $x\geq0$ for all odd $n\geq1$.
      \item[{\rm (iv)}] $\ln h(t,x)$ is convex in $x\geq0$.
\end{enumerate}
\end{lemma}
\begin{proof} (i) We have
\[ h(t,0)=\E[w_p'(t, N(\xi))]=\int^{\infty}_{-\infty}w_p'(t, N(x))dN(x)=\int^1_0w_p'(t, p)dp=w(t,1)-w(t,0)=1.\]

(ii)
For $\varepsilon\not=0$, write
$$\frac{h(t,x+\varepsilon)-h(t,x)}{\varepsilon}=\E\left[w_p'(t, N(\xi))\xi\int_0^1 e^{(x+\theta\varepsilon)\xi}d\theta\right].
$$


For $x\ge 0$, we have
$$0<\int_0^1 e^{(x+\theta\varepsilon)\xi}d\theta\le e^{(x+|\varepsilon|)\xi}+e^{-|\varepsilon|\xi}.$$
Fix $a<0$. For any integer $n\geq 0$, it follows from $h(t,a)<+\infty$ that there exists a constant $C_n>0$ such that
$$|\xi^n|\le C_n(e^{a\xi/2}+e^{\xi}).$$

Now, for any $\varepsilon$ with $|\varepsilon|<\min(-a/2,1)$, we have
\begin{eqnarray*}
 &&|w_p'(t,N(\xi))\xi\int_0^1 e^{(x+\theta\varepsilon)\xi}d\theta|\\
&\le& w_p'(t,N(\xi))|\xi|\left(e^{(x+|\varepsilon|)\xi}+e^{-|\varepsilon|\xi}\right)\\
&\le&C_1 w_p'(t,N(\xi))\left(e^{(x+|\varepsilon|+1)\xi}+e^{(1-|\varepsilon|)\xi}+e^{(x+|\varepsilon|+a/2)\xi}+e^{(a/2-|\varepsilon|)\xi}\right)\\
&<&C_1 w_p'(t,N(\xi))\left(e^{(x-a/2+1)\xi}+e^{(x+1)\xi}+e^{(1+a/2)\xi}+e^{\xi}
+e^{(x+a/2)\xi}+e^{x\xi}+e^{a\xi}+e^{a\xi/2}\right),
\end{eqnarray*}
where to deduce the last inequality we have repeatedly used the fact that $e^y<e^{y_1}+e^{y_2}$ for any $y_2>y_1$ and $y\in [y_1, y_2]$.
From the assumption of the lemma, it follows that the last term of the above is a random variable with
a finite mean; hence
 Lebesgue's dominated convergence theorem yields
$$h'_x(t,x)=\E[w_p'(t,N(\xi))\xi e^{x\xi}].$$
Similarly, we can derive the desired expressions of higher-order derivatives.

(iii) The result is straightforward by (ii).

(iv) It follows from (ii) and the Cauchy--Schwarz inequality that, for $x\geq0$,
\begin{eqnarray*}
h'_x(t,x)&=&\E\left[w_p'(t, N(\xi))\xi e^{x\xi}\right]\\
&=&\E\left[\sqrt{w_p'(t, N(\xi))}\xi e^{x\xi/2} \cdot \sqrt{w_p'(t, N(\xi))} e^{x\xi/2}\right]\\
&\leq& \sqrt{\E\left[w_p'(t, N(\xi))\xi^2 e^{x\xi}\right]}\cdot \sqrt{\E\left[w_p'(t, N(\xi)) e^{x\xi}\right]}\\
&=&\sqrt{h''_x(t,x)}\sqrt{h(t,x)}.
\end{eqnarray*}
Thus
\[ \frac{\partial^2}{\partial x^2}\ln h(t,x)=\frac{h''_x(t,x)h(t,x)-(h'_x(t,x))^2}{h^2(t,x)}\geq0,\]
establishing the desired convexity.
\end{proof}

\section{Equilibrium Strategies}

\subsection{An Ansatz and an ODE}

Our approach to deriving the equilibrium strategies is inspired by an {\it Ansatz} we now make. If the agent was an expected utility maximizer, then
her optimal strategy would be to dynamically replicate the terminal wealth $I(\kappa\rho(T))$ where $\rho(\cdot)$ is the state-price density process
defined as
\begin{equation}\label{rho}
\rho({t}):= \exp\left(-\frac{1}{2}\int_{0}^{t}|\theta(s)|^{2}ds
-\int_{0}^{t}\theta(s)^{\top}dW(s)\right),
\end{equation}
 and $\kappa$ is the Lagrange multiplier for the budget constraint. We conjecture that in the current RDU setting the
terminal wealth from an equilibrium strategy is still of the form $I(\kappa\bar\rho(T))$ with a {\it revised} state-price density process determined by
multiplying the  market price of risk function $\theta(\cdot)$ by a scaling function $\lambda(\cdot)$:
\begin{equation}\label{rhobar}
\bar\rho({t})=\exp\left(-\frac{1}{2}\int_{0}^{t}|\lambda(s)\theta(s)|^{2}ds-\int_{0}^{t}\lambda(s)\theta(s)^{\top}dW(s)\right).
\end{equation}

Since now we have conjectured a {\it specific} form of the desired terminal wealth profile $X(T)$, we will be able to calculate its RDU value along with that of a slightly perturbed
wealth process in the spirit of Definition \ref{equdef}. Then, the equilibrium condition (\ref{key}) will lead to an equation that can be used to identify $\lambda(\cdot)$ as well as to other conditions.

It turns out that the equation to derive $\lambda(\cdot)$ is an ODE, explicitly expressed in the following form:
 \begin{equation}\label{Lambdaeq}
\left\{\begin{array}{l}
\Lambda'(t)=-|\theta(t)|^2\left(\frac{h(t,\sqrt{\Lambda(t)})}{h'_x(t,\sqrt{\Lambda(t)})}\right)^2\Lambda(t),\;\; t\in[0,T),\\
\Lambda(T)=0.
\end{array}\right.
\end{equation}
This is a highly nonlinear ODE that is singular at $T$. The existence of its {\it positive} solutions will be established in Subsection \ref{sec1cond}.

The scaling function $\lambda(\cdot)$  in determining (\ref{rhobar}) is then given by
\begin{equation}\label{thetabar}
\lambda({t}):=\sqrt{-\Lambda'(t)/|\theta(t)|^2}>0,\;\;t\in[0,T),
\end{equation}
where $\Lambda(\cdot)$ is a positive solution of (\ref{Lambdaeq}).

\subsection{Terminal wealth of an equilibrium strategy}

The following result gives a complete solution to our problem by presenting the explicit terminal wealth profile of an equilibrium strategy.

\begin{theorem}\label{mainth}
Assume that  equation (\ref{Lambdaeq}) admits a solution $\Lambda(\cdot)\in C[0,T]\cap C^1[0,T)$ with $\Lambda(t)>0$ $\forall t\in [0, T)$, and that
the following inequality holds  for any $c\in \R$:
\begin{equation}\label{2ndcond}
  \int_{-\infty}^{+\infty} w_p'\left(t,N \left(\frac{c -g(x)}{\sqrt{\Lambda(t)}}\right)\right)
  N'\left(\frac{c -g(x)}{\sqrt{\Lambda(t)}}\right)
      \left(g^\dpm(x)+\frac{c -g(x)}{\Lambda(t)}g'(x)^2 \right) du(x)\geq 0,\;\; a.e. t\in [0, T).
\end{equation}
Moreover, assume there is $\kappa>0$ such that the following holds
\begin{equation}\label{bc}
\E\left[\rho(T)I(\kappa\bar\rho(T))\right]=x_0
\end{equation}
where $x_0>0$ is the initial endowment of the agent at $t=0$.
Then the portfolio replicating the terminal wealth
\begin{equation}\label{terminal}
X({T}):=
I\left(\kappa\bar\rho(T)\right)
\end{equation}
where $\bar\rho(T)$ is determined through (\ref{rhobar}) -- (\ref{thetabar}),
is an equilibrium strategy.
\end{theorem}

\begin{proof}
Denote
$$\cE_{s,t}:=\int_s^t\lambda(v)\theta(v)^\top dW(v)\;\forall 0\leq s\leq t\leq T, \quad f(x):=I(\kappa e^{-\frac{1}{2}\Lambda(0)}e^{-x})=:g^{-1}(x), \quad v(x):=u^{-1}(x),\; x\in\R.$$
It is easy to see that, for any $s\in [0, T)$ and conditional on ${\cal F}_s$, $\cE_{s,t}$ is normal, i.e., $\cE_{s, T}|{\cal F}_s \sim N(0, \Lambda({s}))$, and
(\ref{terminal}) can be rewritten as $X(T)=f(\cE_{0,T})$.

Let $\pi(\cdot)$ be the replicating strategy of $X(T)$, which exists by the market completeness. Moreover, the budget constraint (\ref{bc}) ensures that $\pi(\cdot)$
is an admissible portfolio starting from the initial wealth $x_0$.
The goal is to prove that $\pi(\cdot)$ is an equilibrium strategy.

Fix $t\in [0, T)$. Consider the perturbed strategy described in Subsection 2.3 with
the perturbed  final wealth
$X(T)+k^\top \Delta(t,\varepsilon)$, where $\varepsilon\in(0,T-t)$.  To compute the RDU value of this perturbed strategy, we first calculate, for any $y\in\R $:
\begin{eqnarray*}
  \p_t\left(u(X(T)+k^\top \Delta(t,\varepsilon)) >y\right)&=&\p_t\left(X(T)>v( y)-k^\top \Delta(t,\varepsilon)\right)\\
  &=&\p_t\left(\cE_{t+\varepsilon, T}>g(v(y)-k^\top \Delta(t,\varepsilon))-\cE_{0, t+\varepsilon}\right)\\
  &=&\E_t\left[\p_{t+\varepsilon}\left(\cE_{t+\varepsilon, T}>g(v(y)-k^\top \Delta(t,\varepsilon)\right)-\cE_{0, t+\varepsilon})\right]\\
  &=&\E_t\left[N\left(\frac{\cE_{0, t+\varepsilon}-g(v( y)-k^\top \Delta(t,\varepsilon))}{\sqrt{\Lambda({t+\varepsilon})}}\right)\right],
\end{eqnarray*}
where  $\p_{t}=\p(\cdot|\cF_{t})$ and $ \E_{t}=\E[\cdot|\cF_{t}]$.

Denote  $m(s,y)=v(y)-k^\top \Delta(t,s)$ and $Y(s,y)=\frac{\cE_{0, t+s}-g(m({s,y}))}{\sqrt{\Lambda({t+s})}}$, $s\in[0,\varepsilon)$, $y\in\R$.
Applying It\^o's formula and noting that $\Lambda'(t+s)=-\lambda(t+s)^2|\theta(t+s)|^2$, we derive\footnote{To save space,
we will omit to write out the dependence of $m$ and $Y$ in $y$ from this point of the proof, except when it is important to spell out this
dependence. However, the reader is urged to bear in mind this dependence while reading the proof.}
\begin{eqnarray*}
dY(s)&=
& \left[\frac{Y(s)}{2\Lambda({t+s})}\lambda(t+s)^2|\theta(t+s)|^2
               +\frac{g'(m(s))k^\top \mu({t+s})}{\sqrt{\Lambda({t+s})}}
               -\frac{g^\dpm(m(s))|k^\top \sigma({t+s})|^2}{2\sqrt{\Lambda({t+s})}}\right]ds\\
          &&    +\frac{\lambda({t+s})\theta({t+s})^\top +g'(m(s))k^\top \sigma({t+s})}{\sqrt{\Lambda({t+s})}}dW({t+s}),\;\;s\in[0,\varepsilon).
     \end{eqnarray*}
     Applying It\^o's formula again yields
  \begin{eqnarray*}
dN(Y(s))&=&N'(Y(s))dY(s)+\frac{1}{2\Lambda({t+s})}N^\dpm(Y(s))|\lambda({t+s})\theta({t+s})^\top +g'(m(s))k^\top \sigma({t+s})|^2ds\\
     &=& N'(Y(s))dY(s)-\frac{1}{2\Lambda({t+s})}N'(Y(s))Y(s)|\lambda({t+s})\theta({t+s})^\top +g'(m(s))k^\top \sigma({t+s})|^2ds\\
     &=&A(s)ds+B(s)dW({t+s}),\;\;s\in[0,\varepsilon),
\end{eqnarray*}
where
\begin{equation}\label{aandb}
\begin{array}{lcl}
A(s)&:=&\frac{N'(Y(s))}{\sqrt{\Lambda({t+s})}}\left[g'(m(s))k^\top \mu({t+s})-\frac{g^\dpm(m(s))}{2}|\sigma({t+s})^\top k|^2\right.\\
    &&             \left. -\frac{Y(s) g'(m(s))^2}{2\sqrt{\Lambda({t+s})}}| \sigma({t+s})^\top k |^2-\frac{Y(s) g'(m(s))}{\sqrt{\Lambda({t+s})}}
    k^\top\sigma({t+s})\lambda({t+s})\theta({t+s})\right]\\
  &=&\frac{N'(Y(s))}{\sqrt{\Lambda({t+s})}}
                  \left[-\frac{1}{2}\left(g^\dpm(m(s))+g'(m(s))^2\frac{Y(s)}{\sqrt{\Lambda({t+s})}}\right)|\sigma({t+s})^\top k|^2\right.\\
                  &&             \left.+
                  g'(m(s))\left(1-\frac{Y(s)\lambda({t+s})}{\sqrt{\Lambda({t+s})}}\right)\theta({t+s})^\top \sigma({t+s})^\top k\right],\\
B(s)&:=&\frac{N'(Y(s))}{\sqrt{\Lambda({t+s})}}\left[ \lambda({t+s})\theta({t+s})^\top+g'(m(s))k^\top \sigma({t+s})\right].
\end{array}
\end{equation}
Integrating from $s=0$ to $s=\delta$, where $\delta\in(0,\varepsilon]$, and then taking conditional expectations on the above, we obtain
\begin{equation}\label{nonkey}
\E_t[N(Y(\delta))]=N(Y(0))+\int_0^\delta\E_t [A(s)]ds,
\end{equation}
where we have used the fact that $\int_0^{\cdot}B(s)dW({t+s})$ is a martingale on $[0, \varepsilon]$ for sufficiently small $\varepsilon>0$ 
and hence
\begin{equation}\label{martingale}
\E_t\int_0^{\delta}B(s)dW({t+s})=0,\;\;\delta\in(0,\varepsilon].
\end{equation}
A proof of this martingality will be delayed to  Subsection \ref{martingalecheck}. Now, we have
\begin{eqnarray*}
  &&\limsup_{\epsilon\downarrow 0} \frac{w\left(t,\p_t(u(X(T)+k^\top \Delta(t,\varepsilon)) >y)\right)-
  w\left(t,\p_t(u(X(T)) >y)\right)}{\varepsilon}\\
&=&\limsup_{\varepsilon\downarrow 0} \frac{w(t,\E_t[ N(Y(\varepsilon))])- w(t, N(Y(0)))}{\varepsilon}\\
&=&\limsup_{\varepsilon\downarrow 0} \frac{w\left(t,N(Y(0))+\int_0^\varepsilon\E_t [A(s)]ds\right)- w(t, N(Y(0)))}{\varepsilon}\\
&=&w_p'(t,N(Y(0)))\cdot \limsup_{\varepsilon\downarrow 0} \frac{1}{\varepsilon}\int_0^\varepsilon\E_t [A(s)]ds\\
&=& w_p'(t,N(Y(0)))A(0).
\end{eqnarray*}

Hence (from this point on we will write back the variable $y$)
{\small \begin{equation}\label{cv}
         \begin{array}{lcl}
  && \limsup_{\varepsilon\downarrow 0} \frac{J(X(T)+k^\top\Delta(t,\varepsilon);t,X(t))-J(X(T);t,X(t))}{\varepsilon}\\
  &=& \limsup_{\varepsilon\downarrow 0} \int_{-\infty}^{+\infty } \frac{w(t,\E_t[ N(Y(\varepsilon,y))])- w(t, N(Y(0,y)))}{\varepsilon}dy\\
  &\leq& \int_{-\infty}^{+\infty } \limsup_{\varepsilon\downarrow 0} \frac{w(t,\E_t[ N(Y(\varepsilon,y))])- w(t, N(Y(0,y)))}{\varepsilon}dy\\
    &=&\int_{-\infty}^{+\infty  } w_p'(t,N(Y(0,y)))A(0,y)dy\\
  &=&-|\sigma({t})^\top k|^2 \frac{1}{2\sqrt{\Lambda({t})}}
         \int_{-\infty}^{+\infty  } w_p'(t,N(Y(0,y)))N'(Y(0,y))\left(g^\dpm(m(0,y))+g'(m(0,y))^2\frac{Y(0,y)}{\sqrt{\Lambda({t})}}\right)dy\\
  &&+\theta(t)^\top \sigma(t)^\top k\frac{1}{\sqrt{\Lambda({t})}}
         \int_{-\infty}^{+\infty  } w_p'(t,N(Y(0,y)))N'(Y(0,y))g'(m(0,y))\left(1-\frac{Y(0,y)}{\sqrt{\Lambda({t})}}\lambda({t})\right)dy.
         \end{array}
\end{equation}}
In the above, the first inequality
{\begin{equation}\label{cv2}\begin{array}{rl}
&\limsup_{\varepsilon\downarrow 0} \int_{-\infty}^{+\infty  } \frac{w(t,\E_t[ N(Y(\varepsilon,y))])- w(t, N(Y(0,y)))}{\varepsilon}dy\\
  \leq&\int_{-\infty}^{+\infty  } \limsup_{\varepsilon\downarrow 0} \frac{w(t,\E_t[ N(Y(\varepsilon,y))])- w(t, N(Y(0,y)))}{\varepsilon}dy
  \end{array}
  \end{equation}}
will be proved in Subsection \ref{swap-cond}. 

By Definition \ref{key}, the underlying strategy is an equilibrium if the right hand side of (\ref{cv})  is non-positive for {\it any} $k\in\R^n$.\footnote{Recall that
$t$ is fixed; so any ${\cal F}_t$-measurable random vector $k$ is almost surely  deterministic conditional on ${\cal F}_t$.}
However, the right hand side, being quadratic in  $\sigma(t)^\top k$ while $\sigma(t)$ is invertible,  is non-positive for any $k$
if and only if (noting $\theta(t)\neq 0$)
\begin{equation}\label{1odcon}
  \int_{-\infty}^{+\infty  } w_p'(t,N(Y(0,y)))N'(Y(0,y))g'(m(0,y))\left(1-\frac{Y(0,y)}{\sqrt{\Lambda({t})}}\lambda({t})\right)dy=0,
\end{equation}
and
\begin{equation}\label{2odcon}
\int_{-\infty}^{+\infty } w_p'(t,N(Y(0,y)))N'(Y(0,y))\left(g^\dpm(m(0,y))+g'(m(0,y))^2\frac{Y(0,y)}{\sqrt{\Lambda({t})}}\right)dy\geq0.
\end{equation}

Recall that $m(0,y)=v( y)$ and $Y(0,y)=\frac{\cE_{0,t}-g(v( y))}{\sqrt{\Lambda(t)}}$. By changing variables $x=v(y)$ and
$z=\frac{\cE_{0,t}-g(x)}{\sqrt{\Lambda({t})}}$, the equation (\ref{1odcon}) becomes
\begin{eqnarray*}
  0&=&\int_{-\infty}^{+\infty}w_p'\left(t,N \left(\frac{\cE_{0,t}-g(x)}{\sqrt{\Lambda(t)}}\right)\right)
            N'\left(\frac{\cE_{0,t}-g(x)}{\sqrt{\Lambda(t)}}\right)g'(x)
           \left(1-\frac{\cE_{0,t}-g(x)}{\Lambda(t)}\lambda({t})\right)u'(x)dx\\
&=&-\int_{-\infty}^{+\infty}w_p'\left(t,N \left(\frac{\cE_{0,t}-g(x)}{\sqrt{\Lambda(t)}}\right)\right)
           N'\left(\frac{\cE_{0,t}-g(x)} {\sqrt{\Lambda(t)}}\right)
          \left(1-\frac{\cE_{0,t}-g(x)}{\Lambda(t)}\lambda({t})\right)du'(x)\\
&=&-\int_{ -\infty}^{+\infty} w_p'(t,  N(z)) N'(z)
      \left(1-\frac{z}{\sqrt{\Lambda(t)}}\lambda({t})\right)de^{\sqrt{\Lambda({t})}z-\cE_{0,t}}\\
&=&-e^{-\cE_{0,t}}\int_{ -\infty}^{+\infty}  w_p'(t,  N(z)) N'(z)
      (\sqrt{\Lambda(t)}-z\lambda({t})) e^{\sqrt{\Lambda({t})}z}dz\\
 &=&-e^{-\cE_{0,t}}\E\left[w_p'(t,N(\xi))\left(\sqrt{\Lambda(t)}-\xi\lambda({t})\right) e^{\sqrt{\Lambda({t})}\xi}\right],
\end{eqnarray*}
where $\xi\sim N(0,1)$.
This equation is further equivalent to
$$\sqrt{\Lambda(t)}\E\left[w_p'(t,N(\xi))e^{\sqrt{\Lambda({t})}\xi}\right]
=\lambda({t})\E\left[w_p'(t,N(\xi))\xi e^{\sqrt{\Lambda({t})}\xi}\right],$$
or
$$\lambda({t})=\frac{h(t,\sqrt{\Lambda({t})})}{h'(t,\sqrt{\Lambda({t})})}\sqrt{\Lambda(t)},$$
which holds true by the facts that $\Lambda$ solves equation (\ref{Lambdaeq})
and $\Lambda'({t})=-\lambda(t)^2|\theta(t)|^2$.

Similarly, we can rewrite the inequality (\ref{2odcon}) as
\begin{equation}\label{2ndcond0}
  \int_{-\infty}^{+\infty} \hspace{-0.4cm}w_p'\left(t, N \left(\frac{\cE_{0,t} -g(x)}{\sqrt{\Lambda(t)}}\right)\right)N'\left(\frac{\cE_{0,t} -g(x)}{\sqrt{\Lambda(t)}}\right)
      \left(g^\dpm(x)+g'(x)^2\frac{\cE_{0,t} -g(x)}{\Lambda(t)} \right) du(x)\geq0,
\end{equation}
which is satisfied under (\ref{2ndcond}).
\end{proof}


In the above proof, there are two technical results, the martingale condition (\ref{martingale}) and the inequality (\ref{cv2}), left unproved. We provide
proofs in the next two subsections.

\subsection{The martingale condition (\ref{martingale})}\label{martingalecheck}

\begin{proposition}\label{th4martcond}
The martingale condition (\ref{martingale}) holds for sufficiently small $\varepsilon>0$.
\end{proposition}
\begin{proof}
As before the variable $y$ is suppressed. It suffices to prove that $B(s)$ is locally square integrable at $s=0+$.
Note that  $\lambda(\cdot), \theta(\cdot), \sigma(\cdot)$ are all bounded, and $\frac{1}{\sqrt{\Lambda(t+s)}}$ is
locally bounded at $s=0+$. It thus follows from the expression of $B(\cdot)$ that we only need to
estimate a bound of $N'(Y(s)) g'(m(s))$.

Because $\limsup_{y\rightarrow +\infty}\frac{y^{-\alpha}}{-I'(y)}<+\infty$ (Assumption \ref{uandw}-(i)) and
$\frac{1}{-I'(y)}\equiv -u^\dpm(I(y))$ is increasing in $y$ (Assumption \ref{uandw}-(iii)), we have that there exists $K>0$ such that
\begin{equation}\label{k1}
\frac{1}{-I'(y)}\le K(y^\alpha+1) \;\; \forall\, y>0.
\end{equation}

Recalling that $m(s)=f(\cE_{0, t+s}-\sqrt{\Lambda({t+s})}Y(s))$ and $f'(x)=-\tilde\kappa I'(\tilde\kappa e^{-x})e^{-x}$
where $\tilde \kappa:=\kappa e^{-\frac{1}{2}\Lambda(0)}$, we have
\begin{eqnarray*}
g'(m(s))&=&g'(f(\cE_{0, t+s}-\sqrt{\Lambda({t+s})}Y(s)))\\
&=&\frac{1}{f'(\cE_{0, t+s}-\sqrt{\Lambda({t+s})}Y(s))}\\
&=&\frac{1}{-\tilde\kappa I'(\tilde\kappa e^{\sqrt{\Lambda({t+s})}Y(s)-\cE_{0, t+s}})}e^{\cE_{0, t+s}-\sqrt{\Lambda({t+s})}Y(s)}\\
&\le& \frac{K}{\tilde\kappa}\left(\tilde\kappa^\alpha e^{(\alpha-1)\sqrt{\Lambda({t+s})}Y(s)-(\alpha-1) \cE_{0, t+s}}+e^{\cE_{0, t+s}-\sqrt{\Lambda({t+s})}Y(s)}\right),
\end{eqnarray*}
where the last inequality is due to (\ref{k1}).
However,  $N'(x)e^{\gamma x}$ is bounded in $x\in \R$ for any given $\gamma\in \R$; hence
$$0< N'(Y(s))g'(m(s))\le c_1 e^{-(\alpha-1) \cE_{0, t+s}}+c_2 e^{\cE_{0, t+s}}.$$
The right hand side above is locally square integrable as a process in $s$. The proof is complete. \end{proof}

\subsection{The inequality (\ref{cv2})}\label{swap-cond}

We now prove the inequality (\ref{cv2}). In view of Fatou's lemma, it suffices to show that
the integrand on its left hand side is dominated by an integrable function.
Throughout this subsection, we keep $t\in [0, T)$ fixed.

%

\begin{lemma}\label{Adelta}
For any constant $\gamma\in (0,1]$ and $\varepsilon \in (0, T-t)$, there exists a constant
$c_{1}(\gamma)>0$ such that
$$\frac{|A(\delta,y)|}{N'(Y(\delta,y))^{1-\gamma}}
=\frac{|A(\delta,y)|}{N'(-Y(\delta,y))^{1-\gamma}}\le c_{1}(\gamma) e^{\frac{\gamma}{4\Lambda(t+\delta)}\cE^{2}_{0, t+\delta}} \;\; \forall \delta\in (0,\varepsilon],\;y\in\R.$$
\end{lemma}
\begin{proof}  As before, we drop the $y$ variable to save space in this proof. 
Since $g(x)=l(x)+\ln \kappa -\frac{1}{2}\Lambda(0)$, Assumption \ref{uandw}-(ii) is satisfied with $l(\cdot)$ replaced by $g(\cdot)$ (and the constants $a$ and $b$ properly modified). Hence, 
\begin{eqnarray*}
\frac{|A(\delta)|}{N'(Y(\delta))^{1-\gamma}}
&\le &c_{2}(\gamma)N'(Y(\delta))^{\gamma}(e^{a|g(m(\delta))|+b}+|Y(\delta)|e^{2a|g(m(\delta))|+2b})\\
&= &\frac{c_{2}(\gamma)}{\sqrt{2\pi} }e^{-\frac{\gamma}{2}Y(\delta)^{2}}(e^{a|g(m(\delta))|+b}+|Y(\delta)|e^{2a|g(m(\delta))|+2b})\\
&\le&c_{3}(\gamma)e^{-\frac{\gamma}{4}Y(\delta)^{2}}e^{2a|g(m(\delta))|+2b}\\
&=&c_{3}(\gamma)e^{-\frac{\gamma}{4\Lambda(t+\delta)}(\cE_{0, t+\delta}-g(m(\delta)))^{2}}e^{2a|g(m(\delta))|+2b},
\end{eqnarray*}
where $c_{2}(\gamma)>0$ and $c_{3}(\gamma)>0$ are suitable constants.
Making use of the general inequality
$(x+y)^{2}\ge  \frac{1}{2}x^{2}- y^{2}$, we deduce
\begin{eqnarray}
\frac{|A(\delta)|}{N'(Y(\delta))^{1-\gamma}}
&\le&c_{3}(\gamma)e^{-\frac{\gamma}{4\Lambda(t+\delta)}(\cE_{0, t+\delta}-g(m(\delta)))^{2}}e^{2a|g(m(\delta))|+2b}\nonumber\\
&\le&c_{3}(\gamma)e^{\frac{ \gamma}{4\Lambda(t+\delta)}\cE_{0, t+\delta}^{2}}
e^{-\frac{ \gamma}{8\Lambda(t+\delta)}g(m(t+\delta))^{2}+2a|g(m(t+\delta))|+2b}\nonumber\\
&\le&c_{3}(\gamma)e^{n_{1}(\gamma)}e^{\frac{ \gamma}{4\Lambda(t+\delta)}\cE_{0, t+\delta}^{2}},
\label{estofAdel}
\end{eqnarray}
where $n_{1}(\gamma):=\max_{x\in \R} \{-\frac{ \gamma}{8\Lambda_{t+\delta}}x^{2}+2a|x|+2b\}<+\infty$. This completes the proof.
\end{proof}

\begin{lemma}\label{EA/EN}
For any $\gamma\in (0,1)$, there exists $\varepsilon\in (0, T-t)$ and a constant $M_{1}(\gamma)<+\infty$  such that
$$\frac{\E_{t}|A(\delta,y)|}{[\E_{t}N(Y(\delta,y))]^{\gamma}}\le M_{1}(\gamma) \quad \forall\, \delta\in (0, \epsilon],\;y\in\R.$$
\end{lemma}
\begin{proof} Again we omit to write out $y$. Denote $m=1/\gamma$ and $n=\frac{m}{m-1}$. Then Cauchy--Schwarz's inequality yields
\begin{eqnarray*}
\frac{\E_{t}|A(\delta)|}{[\E_{t}N(Y(\delta))]^{\gamma}}
&\equiv &\frac{\E_{t}|A(\delta)|}{[\E_{t}N(Y(\delta))]^{1/m}}\\
&\le & \left(\E_t\left[|A(\delta)|^nN(Y(\delta))^{-n/m}\right]\right)^{1/n}.
\end{eqnarray*}

Write
\begin{eqnarray*}
|A(\delta)|^n N(Y(\delta))^{-n/m}&=&\left(\frac{N'(Y(\delta))}{N(Y(\delta))}\right)^{n/m}
\left(\frac{|A(\delta)|}{N'(Y(\delta))^{1/m}}\right)^n\\
&=&\left(\frac{N'(Y(\delta))}{N(Y(\delta))}\right)^{n/m}
\left(\frac{|A(\delta)|}{N'(Y(\delta))^{(\gamma+1)/2}}\right)^n N'(Y(\delta))^{n(1-\gamma)/2}.
\end{eqnarray*}
Noting $\frac{N'(x)}{N(x)}\le C(|x|+1)$ $\forall x\in\R$ for some $C>0$, we conclude that  the last term above converges to $0$ faster than
the first term going to $+\infty$ when $Y(\delta)$ goes to $+\infty$; hence we can find a bound $M_{2}(\gamma)>0$ such that
\begin{eqnarray*}
|A(\delta)|^n N(Y(\delta))^{-n/m}&\le& M_{2}(\gamma)\left(\frac{|A(\delta)|}{N'(Y(\delta))^{(\gamma+1)/2}}\right)^n\\
&\le&M_{2}(\gamma)c_{1}((1-\gamma)/2)^{n} e^{\frac{1}{8\Lambda(t+\delta)}\cE_{0, t+\delta}^{2}},
\end{eqnarray*}
where the last inequality is due to Lemma \ref{Adelta}.
When $\varepsilon>0$ is small enough, there exists $M_3(\gamma)<\infty$ such that  $\E_{t} [e^{\frac{1}{8\Lambda(t+\delta)}\cE_{0, t+\delta}^{2}}]<M_3(\gamma)$ for any $\delta\in (0, \epsilon]$. This leads to the desired inequality.
\end{proof}

\begin{proposition}\label{y-infty}
We have the following conclusions:
\begin{itemize}
\item[{\rm (i)}] For any $\gamma>0$, there are a sufficiently small $\varepsilon\in (0, T-t)$ and a function $H_{1}(\cdot;\gamma)$ with $\int_{-\infty}^{\underline y(\gamma)}H_{1}(y;\gamma)dy<+\infty$ for some $\underline y(\gamma)<0$, such that when $\delta\in (0, \varepsilon)$, we have
$$(\E_{t}|A({\delta})|)^{\gamma}\le H_{1}(y;\gamma)\;\;\forall y\in (-\infty, \underline y(\gamma)].
$$
\item[{\rm (ii)}] For any $\gamma>0$, there are a sufficiently small $\varepsilon\in (0, T-t)$ and a function $H_{2}(\cdot;\gamma)$ with $\int_{\bar y(\gamma)}^{+\infty}H_{2}(y;\gamma)dy<+\infty$ for some $\bar y(\gamma)>0$, such that when $\delta\in (0, \varepsilon)$, we have
$$(\E_{t}|A({\delta})|)^{\gamma}\le H_{2}(y;\gamma)\;\;\forall y\in [\bar y(\gamma),+\infty).$$
\end{itemize}
\end{proposition}

\begin{proof} We prove only (i), while (ii) being similar.   Applying $\gamma=1$ to the second inequality of (\ref{estofAdel}), we get
\begin{eqnarray*}
|A(\delta)|&\le&c_{3}(1)e^{\frac{1}{4\Lambda(t+\delta)}\cE_{0, t+\delta}^{2}}
e^{-\frac{1 }{8\Lambda(t+\delta)}g(m(\delta))^{2}+2a|g(m(\delta))|+2b}.
\end{eqnarray*}
Define $n_{2}:=\max_{x\in {\mathbb R}} \{-\frac{ 1}{16\Lambda(t+\delta)}x^{2}+2a|x|+2b\}<+\infty$.
Then
$$|A(\delta)| \le c_{3}(1)e^{n_{2}}e^{\frac{1}{4\Lambda(t+\delta)}\cE_{0, t+\delta}^{2}}
e^{-\frac{1 }{16\Lambda(t+\delta)}g(m(\delta))^{2}}.
$$
So
\begin{eqnarray*}
\E_{t}|A(\delta)|&\le&c_{3}(1)e^{n_{2}}\E_{t}[e^{\frac{1}{4\Lambda(t+\delta)}\cE_{0, t+\delta}^{2}}
e^{-\frac{1 }{16\Lambda(t+\delta)}g(m(\delta))^{2}}]\\
&\le&c_{3}(1)e^{n_{2}}\sqrt{\E_{t}e^{\frac{1}{2\Lambda(t+\delta)}\cE_{0, t+\delta}^{2}}}
\sqrt{\E_{t}e^{-\frac{1 }{8\Lambda(t+\delta)}g(m(\delta))^{2}}}.
\end{eqnarray*}

Set $n_{3}:=c_{3}(1)^{2} e^{2n_{2}}\E_{t}[e^{\frac{1}{2\Lambda_{t+\delta}}\ \cE_{0,t+\delta}^{2}}]$ which is a finite constant  when $\varepsilon$ (and hence $\delta$) is sufficiently small, 
and $\xi:=k^{\top}\Delta(t,\delta)\sim N(0, \eta^{2})$ conditional on
${\cal F}_t$, where $\eta^2\equiv \eta(\delta)^2={\rm Var}(\xi)\le |k|^{2}\int_{t}^{t+\delta}|\sigma(s)|^{2}ds$.  Then
\begin{eqnarray*}
(\E_{t}|A(\delta)|)^{2}&\le& n_{3}\E_{t}[e^{-\frac{1 }{8\Lambda(t+\delta)}g(m(\delta))^{2}}]\\
&=&n_{3}\E_{t}[e^{-\frac{1 }{8\Lambda(t+\delta)}\ln^{2}u'(v(y)-\xi)}]\\
&=&\frac{n_{3}}{\sqrt{2\pi}}  \int_{-\infty}^{+\infty}e^{-\frac{1}{8\Lambda(t+\delta)} \ln^{2}u'(v(y)+\eta z)}e^{-z^{2}/2}dz\\
&=&\frac{n_{3}}{\sqrt{2\pi}} \int_{-\infty}^{+\infty}e^{-\frac{1}{8\Lambda(t+\delta)}
             \ln^{2}u'(\eta\tilde z)}e^{-(\tilde z-v(y)/\eta)^{2}/2}d\tilde z\\
&=&\frac{n_{3}}{\sqrt{2\pi}} \left[\int_{-\infty}^{v(y)/(2\eta)}+\int_{v(y)/(2\eta)}^{0}+\int_{0}^{+\infty}\right]
     e^{-\frac{1}{8\Lambda(t+\delta)} \ln^{2}u'(\eta\tilde z)}e^{-(\tilde z-v(y)/\eta)^{2}/2}d\tilde z.\\
\end{eqnarray*}

We now find an integrable bound (as a function of $y$) for  each of the three integrals in the above.

For the first integral, take $y_{0}<0$  small enough such that $v(y_0)<0$ and $u'(v(y_{0})/2)>1$.
Then $v(y)<0$ and $u'(v(y)/2)\geq u'(v(y_{0})/2)>1$ $\forall y\le y_{0}$.  Thus
\[
\int_{-\infty}^{v(y)/(2\eta)}e^{-\frac{1}{8\Lambda(t+\delta)} \ln^{2}u'(\eta\tilde z)}e^{-(\tilde z-v(y)/\eta)^{2}/2}d\tilde z
<\sqrt{2\pi}e^{-\frac{1}{8\Lambda(t+\delta)} \ln^{2}u'(v(y)/2)}
\le c_{4}(a_{1})u'(v(y)/2)^{-a_{1}},\;\;y\le y_{0},
\]
for any $a_{1}>0$ and some $c_{4}(a_{1})<+\infty$.\footnote{This is based on the fact (which can be easily shown) that given $\alpha>0$, for any $a>0$ there is a constant
$c(a)>0$ such that $e^{-\alpha \ln^{2}x}\leq c(a)x^{-a}$ $\forall x>1$.}

For the second integral, we have
\[
\int_{v(y)/(2\eta)}^{0}e^{-\frac{1}{8\Lambda(t+\delta)} \ln^{2}u'(\eta\tilde z)}e^{-(\tilde z-v(y)/\eta)^{2}/2}d\tilde z
\le e^{-v(y)^{2}/(8\eta^{2})}|v(y)/(2\eta)|\le c_{5}e^{-v(y)^{2}/(9\eta^{2})}
\]
for some constant $c_5>0$.

For the last integral, we get
\begin{eqnarray*}
&&\int_{0}^{+\infty}e^{-\frac{1}{8\Lambda(t+\delta)} \ln^{2}u'(\eta\tilde z)}e^{-(\tilde z-v(y)/\eta)^{2}/2}d\tilde z\\
&=&\int_{0}^{+\infty}e^{-\frac{1}{8\Lambda(t+\delta)} \ln^{2}u'(\eta\tilde z)}e^{-\tilde z^{2}/2}e^{\tilde zv(y)/\eta}e^{-v(y)^{2}/(2\eta^{2})}d\tilde z\\
&<&e^{-v(y)^{2}/(2\eta^{2})}\int_{0}^{+\infty}e^{-\tilde z^{2}/2}d\tilde z=\sqrt{\frac{\pi}{2}}e^{-v(y)^{2}/(2\eta^{2})}.
\end{eqnarray*}

Combining the three integrals, we conclude that,
for any $a_{1}>0, a_{2}>0$,
when $\varepsilon\in (0, T-t)$ is small enough and $y_1<0$ with $v(y_1)$ sufficiently negative, there exists
constant $c_{6}>0$ such that   for any $\delta\in (0, \varepsilon)$,
$$(\E_{t}|A(\delta)|)^{2}\le c_{6}[u'(v(y))^{-a_{1}}+e^{-a_{2}v(y)^{2}}]:=H_{1}(y;\gamma)^{2/\gamma}\;\;\forall y\le y_1.$$

Finally, for any $\gamma>0$, we can take suitable $a_{2}$ and $ a_{3}$ (which may depend on $\gamma$) such that,
in view of  Assumption \ref{swappingextra}, $\int_{-\infty}^{\underline y(\gamma)}H_{1}(y)dy<+\infty$,
where $\underline y(\gamma)$ is sufficiently negative.
\end{proof}

\begin{theorem}\label{swapok}
There exists an integrable function $ H(\cdot)$ such that when $\varepsilon>0$ is sufficiently small, it holds that
$$ \left| \frac{w(t,\E_t[ N(Y(\varepsilon,y))])- w(t, N(Y(0,y)))}{\varepsilon}\right|\le H(y)\;\;\forall y\in \R.$$
\end{theorem}
\begin{proof} Fix $\varepsilon>0$. For any $y\in \R$, by the mean-value theorem, there exists $\delta\in [0, \epsilon]$ (which depends on $t$ and $y$) such that

\begin{eqnarray*}
\left| \frac{w(t,\E_t[ N(Y(\varepsilon,y))])- w(t, N(Y(0,y)))}{\varepsilon}\right|
&=& \left|w'_p(t,\E_t[ N(Y(\delta,y)])\E_t[A(\delta)]\right|\\
&\le&w'_p(t,\E_t[ N(Y(\delta,y)])\E_t[|A(\delta)|],
\end{eqnarray*}
where we have used the fact that, by virtue of (\ref{nonkey}), $\frac{d}{d\delta}  \E_t[ N(Y(\delta,y))] =\E_t[A(\delta)].$
By Assumption 3.1-(iv), we have
\begin{eqnarray*}
\left| \frac{w(t,\E_t[ N(Y(\varepsilon,y))])- w(t, N(Y(0,y)))}{\varepsilon}\right|
&\le&c \frac{\E_t[|A(\delta)|]}{(\E_t[ N(Y(\delta,y)])^{-m}}+c\frac{\E_t[|A_\delta|]}{(1-\E_t[ N(Y(\delta,y)])^{-m}}
\end{eqnarray*}
for some $c>0$ and $m\in(-1,0)$.

Recall that $\E_t[ N(Y(\delta,y)]=\p_{t}\left(u(X({T})+k^{\top}\Delta(t,\delta))>y\right)$.
Choose  $y_{1}$ and $y_2$ such that $\E_t[ N(Y(\delta,y_1)]<1/2$ and
$\E_t[ N(Y(-\delta,y_2)]<1/2$
for all $\delta\in (0, \varepsilon]$.

Using notations in Lemma \ref{EA/EN} and Proposition \ref{y-infty}, denote
 $$\underline y:=\underline y((1+m)/2)) \wedge y_{1}, \quad \bar y:=\bar y((1+m)/2) \vee y_{2},$$
and
\begin{eqnarray*}
\hat H_{1}(y):=2cH_{1}(y; (1+m)/2) M_{1}((2m)/(m-1))^{(1-m)/2},\\
\hat H_{2}(y):=2cH_{2}(y; (1+m)/2) M_{1}((2m)/(m-1))^{(1-m)/2}.
\end{eqnarray*}
Then:
\begin{itemize}
\item [(i)] For any $y<\underline y$, we deduce
\begin{eqnarray*}
\left| \frac{w(t,\E_t[ N(Y(\varepsilon,y))])- w(t, N(Y(0,y)))}{\varepsilon}\right|
&\le&2c\frac{\E_t[|A(\delta)|]}{(\E_t[ N(Y(\delta,y)])^{-m}}\\
&\le&2c(\E_t[|A(\delta)|])^{(1+m)/2}\frac{(\E_t[|A(\delta)|])^{(1-m)/2}}{(\E_t[ N(Y(\delta,y)])^{-m}}\\
&\le&2cH_{1}(y; (1+m)/2) M_{1}((2m)/(m-1))^{(1-m)/2}\\
&=&\hat H_{1}(y).
\end{eqnarray*}

\item [(ii)] Similarly, for any $y>\bar y$, we have $\left| \frac{w(t,\E_t[ N(Y(\varepsilon,y))])- w(t, N(Y(0,y)))}{\varepsilon}\right|\le \hat H_{2}(y)$.
\item [(iii)] Finally, $\left| \frac{w(t,\E_t[ N(Y(\epsilon,y))])- w(t, N(Y(0,y)))}{\epsilon}\right|\le 2cM_{1}(-m)$ when $y\in [\underline y, \bar y]$.
\end{itemize}
Then $H(\cdot)$, where
$H(y)=\hat H_{1}(y)\id_{y< \underline y} +2cM_{1}(-m)+\hat H_{2}(y)\id_{y\ge \bar y}$, $y\in\R$, is the desired function.
\end{proof}

\section{Sufficient Conditions}
Theorem \ref{mainth} holds under three major assumptions: the existence of a positive constant $\kappa$ such that the budget constraint (\ref{bc}) holds,
the existence of a positive solution to the ODE (\ref{Lambdaeq}), and the validity of the
inequality (\ref{2ndcond}).
It is  hard to verify these assumptions (and hence hard to tell if they are reasonable) because they are not imposed directly on the model primitives. This section explores equivalent or at least sufficient conditions, that are shown to be reasonable economically, under which these three assumptions are satisfied respectively.

\subsection{Budget constraint (\ref{bc})} \label{bccond}

The equation (\ref{bc}) is {\it not} a ``standard" budget constraint appearing in the classical Merton problem because $\rho(T)$ and $\bar \rho(T)$ are in general different. However, we have the following result.

\begin{theorem}\label{bcequiv}
The following two statements are equivalent:
\begin{enumerate}
 \item[{\rm (i)}] There exists $\kappa>0$ such that (\ref{bc}) holds.
 \item[{\rm (ii)}]  There exists $\bar\kappa>0$ such that
\begin{equation}\label{bcbar}
\E[\bar\rho(T)I(\bar\kappa \bar\rho(T))]=x_0.
\end{equation}
\end{enumerate}
\end{theorem}

\begin{proof} Let $\Q$ and $\bar\Q$ be respectively the equivalent martingale measures corresponding to $\rho(T)$ and $\bar\rho(T)$; namely,
\[ \frac{d\Q}{d\p}=\rho(T),\;\;\frac{d\bar\Q}{d\p}=\bar\rho(T).\]
By Girsanov's theorem, $\tilde W(\cdot)$ and $\bar W(\cdot)$ are respectively Brownian motions under $\Q$ and $\bar\Q$, where
\[ \tilde W(t):=W(t)+\int_0^t\theta(s)ds,\;\;\bar W(t):=W(t)+\int_0^t\lambda(s)\theta(s)ds,\;\;t\in[0,T].
\]
Now, we have
\begin{eqnarray*}
\E[\rho(T)I(\kappa \bar\rho(T))]&=&\E^\Q\left[I(\kappa e^{-\frac{1}{2}\int_{0}^{T}|\lambda(s)\theta(s)|^{2}ds
-\int_{0}^{T}\lambda(s)\theta(s)^{\top}dW(s)})\right]\\
&=&\E^\Q\left[I(\kappa e^{-\frac{1}{2}\int_{0}^{T}|\lambda(s)\theta(s)|^{2}ds
+\int_{0}^{T}\lambda(s)|\theta(s)|^{2}ds
-\int_{0}^{T}\lambda(s)\theta(s)^{\top}d\tilde W(s)})\right]\\
&=&\E\left[I(\kappa e^{-\frac{1}{2}\int_{0}^{T}|\lambda(s)\theta(s)|^{2}ds
+\int_{0}^{T}\lambda(s)|\theta(s)|^{2}ds
-\int_{0}^{T}\lambda(s)\theta(s)^{\top}d W(s)})\right]\\
&=&\E^{\bar\Q}\left[I(\kappa e^{-\frac{1}{2}\int_{0}^{T}|\lambda(s)\theta(s)|^{2}ds
+\int_{0}^{T}\lambda(s)|\theta(s)|^{2}ds
-\int_{0}^{T}\lambda(s)\theta(s)^{\top}d \bar W(s)})\right]\\
&=&\E\left[\bar\rho(T)I(\kappa e^{-\frac{1}{2}\int_{0}^{T}|\lambda(s)\theta(s)|^{2}ds
+\int_{0}^{T}\lambda(s)|\theta(s)|^{2}ds
-\int_{0}^{T}\lambda(s)\theta(s)^{\top}d \bar W(s)})\right]\\
&=&\E\left[\bar\rho(T)I\left(\kappa e^{\int_{0}^{T}\lambda(s)(1-\lambda(s))|\theta(s)|^{2}ds}\bar\rho(T)\right)\right].
\end{eqnarray*}
This establishes the desired equivalence with $\bar\kappa=\kappa e^{\int_{0}^{T}\lambda(s)(1-\lambda(s))|\theta(s)|^{2}ds}$.
\end{proof}

So the assumption in Theorem \ref{mainth} regarding the existence of a positive constant $\kappa$ satisfying  (\ref{bc}) boils down to the familiar condition (\ref{bcequiv}).
The latter condition is standard in the classical Merton problem in which the pricing kernel is $\bar\rho(T)$ or, equivalently, the market price of risk process is
$\lambda(\cdot)\theta(\cdot)$. Note that because the probability weighting function $w$ has been embedded into $\lambda(\cdot)$, the existence of a positive $\bar\kappa$ satisfying (\ref{bcequiv}) becomes a condition on the utility function $u$ only, which is satisfied by, say, the exponential utility.\footnote{Indeed, if this condition
fails, than it is usually an indication that the original problem is not well-posed and/or an optimal solution is not attainable; see Jin et al
(2008) for a detailed analysis on this constraint.}

More important, with Theorem \ref{bcequiv}, Theorem \ref{mainth} shows that
the investment behavior of the sophisticated RDU agent is {\it indistinguishable} from
an EUT maximizer in a market where the  market price of risk is revised from $\theta(\cdot)$ to $\lambda(\cdot)\theta(\cdot)$. This finding may have important economic implications especially in the study of market equilibria.

\subsection{Existence of positive solutions to (\ref{Lambdaeq})} \label{sec1cond}

The main result of this paper depends crucially on the existence of a positive solution to the
ODE (\ref{Lambdaeq}). Note that this equation is highly nonlinear, and singular at
$t=T$ in that the denominator of the right hand side of the equation is 0 at $T$.
In this subsection we provide conditions under which (\ref{Lambdaeq}) admits positive solutions,
by applying a general existence result for a class of ODEs with singular initial/terminal values (see Appendix A).

Setting $y(t)=\Lambda (T-t)$, (\ref{Lambdaeq}) is equivalent to
\begin{equation}\label{eqy}
\left\{\begin{array}{l}
y'(t)=|\theta({T-t})|^2\left(\frac{h(T-t,\sqrt{y(t)})}{h_x'(T-t,\sqrt{y(t)})}\right)^2 y(t), \quad t\in (0, T],\\
y(0)=0.
\end{array}\right.
\end{equation}
In the rest of this subsection we study  equation (\ref{eqy}) instead of (\ref{Lambdaeq}). The key idea is to first establish the local existence in the right neighborhood of $t=0$, and then
extend it globally to the whole time interval $[0,T]$.

We introduce the following assumption on the function $h(\cdot,\cdot)$ (which depends directly on the probability weighting function) and on the market represented by $\theta(\cdot)$:
\begin{assump}\label{h13}
\begin{enumerate}
 \item[{\rm (i)}] $h_x'(t,0)\ge 0$ and $h_{x}'''(t,0)\ge 0$ $\forall t\in [0, T]$.
\item[{\rm (ii)}]
$\limsup_{t\uparrow T} \frac{h_x'(t,0)}{\sqrt{|\theta(t)|^2 (T-t)}}<1$ and $\liminf_{t\uparrow T}|\theta(t)|^2>0$.
\item[{\rm (iii)}] $\sup_{t\in[0,T]}h(t,1)<+\infty$ and
$\limsup_{t\uparrow T}h_{x}'' (t,1)< +\infty$.
\item[{\rm (iv)}] $\inf_{t\in [0, T]}h_{x}''(t,0)>0$.
\end{enumerate}
 \end{assump}

Assumption \ref{h13}-(iii) and -(iv) are mild. Assumption \ref{h13}-(i) can be relaxed with a more subtle analysis than the one to be given below; although we will not pursue in that direction.
The first part of Assumption \ref{h13}-(ii) is the most important of all, which regulates how the probability weighting function $w(t,\cdot)$ should behave, given the market,
when $t$  is sufficiently close to the terminal time $T$.\footnote{We believe that this is a distinctive feature  of the continuous-time setting. In the discrete-time case, there is no infinitesimal issue of the weighting functions close to the terminal time.}
 Mathematically, this terminal behavior of weighting functions is
translated into the singularity of the ODE (\ref{Lambdaeq}) at $T$, which is why Assumption \ref{h13} is needed for a proof of the existence of (\ref{Lambdaeq}).
Luckily, we will show in Appendix B that all the parts of Assumption \ref{h13} are satisfied by a family of time-varying
Tversky--Kahnamen's weighting functions.

We first strengthen Lemma \ref{hregularity} under Assumption \ref{h13}-(i).

\begin{lemma}\label{star}
Under Assumption \ref{h13}-(i) in addition to the same assumption of Lemma \ref{hregularity}, for any $t\in [0, T]$, $h(t,\cdot)$ has the following properties:
\begin{enumerate}
 \item[{\rm (i)}] $h(t,x)$ and $h^\dpm (t,x)$ are both increasing in $x\geq0$, and $h'_x(t,x)$ is convex in $x\geq0$.
\item [{\rm (ii)}] $xh^\dpm_x(t,0)\leq h'_x(t,x)\leq h'_x(t,0)+xh^\dpm_x(t,x)$ $\forall x\ge 0$.
\end{enumerate}
\end{lemma}

\begin{proof}
(i) It follows from Lemma \ref{hregularity}-(iii) that $h'_x(t,x)$ and $h^{(3)}_x(t,x)$ are both increasing in $x\ge 0$. Assumption \ref{h13}-(i) then leads to the desired results readily.

(ii) Applying the Taylor expansion, we have for any $x>0$, there exists $\zeta\in [0, x]$ such that
$$h'_x(t,x)=h'_x(t,0)+xh^\dpm_x(t,0)+\frac{1}{2}x^2h^{(3)}_x(t,\zeta)\ge xh^\dpm_x(t,x).$$
On the other hand, the convexity of  $h'_x(t,x)$ in $x\ge 0 $ implies that
$$h'_x(t,x)\le h'_x(t,0)+xh^\dpm_x(t,x).$$
The proof is complete.
\end{proof}


\begin{lemma}\label{k-cond} Under Assumption \ref{h13},
there exist $k_1>0$ greater than any given number and $\delta_1>0$ such that
\begin{equation}\label{k0}
 |\theta({T-t})|^2\left(\frac{ \sqrt{t}h(T-t,\sqrt{k_1t})}{h'_x(T-t,\sqrt{k_1t})}\right)^{2}<1 \;\; \forall t\in (0, \delta_1].
\end{equation}
Moreover, there exist $k_2>0$ less than any given positive number and $\delta_2>0$ such that
\begin{equation}\label{k2}
|\theta({T-t})|^2\left(\frac{ \sqrt{t}h(T-t,\sqrt{k_2t})}{h'_x(T-t,\sqrt{k_2t})}\right)^{2}>1\;\; \forall t\in (0, \delta_2].
\end{equation}
\end{lemma}
\begin{proof}
By Lemma \ref{star}-(i), $h(T-t, x)$ is increasing when $x\ge0$; hence
\begin{equation}\label{k11}
\limsup_{t\downarrow 0} h(T-t,\sqrt{kt})\le \limsup_{t\downarrow 0} h(T-t,1)<+\infty,
\end{equation}
where the finiteness is due to  Assumption \ref{h13}-(iii).
On the other hand, it follows from Lemma \ref{star}-(ii) that, for any $k>0$,
\begin{equation}\label{k12}
\frac{h'_x(T-t,\sqrt{kt})}{ \sqrt{t}}\ge \sqrt{k}h_x^{\dpm}(T-t,0)\ge \sqrt{k}\inf_{s\in [0, T]}h_x^{\dpm}(s,0)>0,
\end{equation}
where the last inequality is due to Assumption \ref{h13}-(iv). Combining (\ref{k11}) and (\ref{k12}) and noting
the boundedness of $\theta(\cdot)$,  we conclude that there is $k=k_1$ greater than any given number such that (\ref{k0}) is satisfied.

Next, by Lemma \ref{star}-(ii),  we deduce that for sufficiently small $k>0$,
 \begin{eqnarray*}
  \limsup_{t\downarrow 0}\frac{h'_x(T-t,\sqrt{kt})}{ \sqrt{|\theta({T-t})|^2t}}
  &\le& \limsup_{t\downarrow 0}\left\{\frac{h'_x(T-t,0)}{ \sqrt{|\theta({T-t})|^2t}}
            +\sqrt{k}\frac{h^\dpm_x(T-t,\sqrt{kt})}{\sqrt{|\theta({T-t})|^2}}\right\}\\
  &\le&\limsup_{t\downarrow 0}\frac{h'_x(T-t,0)}{ \sqrt{|\theta({T-t})|^2t}}
             +\sqrt{k}\limsup_{t\downarrow 0}\frac{h_x^\dpm(T-t,1)}{\sqrt{|\theta({T-t})|^2}}\\
             &<&1,
 \end{eqnarray*}
 where the last inequality follows from
Assumption \ref{h13}-(ii) and -(iii). Since $h(T-t,\sqrt{kt})\ge h(T-t,0)=1$, we conclude that there is $k=k_2$ less
than any given positive number so that (\ref{k2}) is satisfied.
\end{proof}

The following establishes the local existence of the ODE (\ref{eqy}).
\begin{proposition}\label{localex4eqy}
Under 
Assumption \ref{h13},  equation (\ref{eqy}) admits a solution
$y(\cdot)\in C[0,\delta]\cap C^1(0,\delta]$ on a  time interval $[0, \delta]$ for some $\delta>0$, with $y(t)>0$
$\forall t\in (0, \delta]$.
\end{proposition}

\begin{proof}
 Fix $k_1>0$ and $0<k_2<k_1$ as in Lemma \ref{k-cond}, and take
$\delta=\delta_1\wedge\delta_2$. Then it is easy to check that
$\beta(t)=k_1t$ and $\alpha(t)=k_2t$ satisfy all the requirements in  Theorem \ref{localex} on
$(0,  \delta]$; hence the result.
\end{proof}

Next we extend the local solution $y(\cdot)$ obtained in Proposition \ref{localex4eqy} to the whole time interval $[0,T]$.
Denote $y(\delta)=x_1>0$. Without loss of generality, we assume $x_1<1$ (otherwise, we can make $\delta$ closer to $0$ to reduce $x_1$).

\begin{proposition}\label{globalex4eqy} Under Assumption \ref{h13}-(i) and (iv), for any $\delta>0$ and $0<x_1<1$,
the equation
\begin{equation}\label{eqy2}
\left\{\begin{array}{l}
y'(t)=|\theta({T-t})|^2\left(\frac{h(T-t,\sqrt{y(t)})}{h'_x(T-t,\sqrt{y(t)})}\right)^2 y(t), \quad t\in (\delta, T],\\
y(\delta)=x_1
\end{array}\right.
\end{equation}
has a solution $y(\cdot)\in C^1[\delta, T]$ with $y(t)>0$ $\forall t\in [\delta, T]$.
\end{proposition}
\begin{proof}
Since $\ln h(x,t)$ is convex in $x\ge0$ (by Lemma \ref{hregularity}-(iv)),
we have
 $$\frac{h'_x(t,x)}{h(t,x)}\ge \frac{h'_x(t,\sqrt{x_1})}{h(t,\sqrt{x_1})}\ge \sqrt{x_1} \frac{h^\dpm_x(t,0)}{h(t,\sqrt{x_1})}\ge
 \sqrt{x_1} \frac{h^\dpm_x(t,0)}{h(t,1)}\ge c_1>0
 \;\;\forall t\in [0, T],\;x\ge \sqrt{x_1},$$
 where the second inequality is due to Lemma \ref{star}-(ii) and $c_1$ (which may depend on $x_1$) is a constant arising from Assumption \ref{h13}-(iii) and -(iv). Hence  for
any $y\ge x_1$, we have
\begin{eqnarray*}\label{control}
|\theta({T-t})|^2\left(\frac{h(T-t,\sqrt{y})}{h'_x(T-t,\sqrt{y})}\right)^2 y
&\le&c_2 y,
\end{eqnarray*}
where $c_2>0$ is a constant depending on $c_1$ and the bound of  $|\theta(\cdot)|^2$.

Denote $c_3=x_1e^{c_2T}<+\infty$, and a truncation function $r(y)=(y\vee x_1)\wedge c_3$ for $y\ge0$.  Consider the ODE
\begin{equation}\label{auxeqy}
\left\{\begin{array}{l}
y'(t)=f(t,y(t)), \quad t\in (\delta, T],\\
y(\delta)=x_1
\end{array}\right.
\end{equation}
where
$$f(t,y):=|\theta({T-t})|^2\left(\frac{h(T-t,\sqrt{r(y)})}
{h_x'(T-t,\sqrt{r(y)})}\right)^2r(y).$$
It is easy to show that $f(\cdot,\cdot)$ satisfies the conditions in Theorem \ref{globalex} on $(\delta, T]$;
hence  (\ref{control}) admits
a solution $y(\cdot)$. Moreover, since $y'(t)\geq0$ we have $y(t)\ge x_1$ $\forall t\in [\delta, T]$.

Take $t_1:=\inf\{t\in (\delta, T]: y(t)\ge c_3\}\wedge T$. If $t_1<T$, then $y(t_1)=c_3$
and $y(t)<c_3$ $\forall t<t_1$. Consequently
$y'(t)=|\theta({T-t})|^2\left(\frac{h(T-t,\sqrt{y(t)})}{h'_x(T-t,\sqrt{y(t)})}\right)^2 y(t)\le c_2 y(t)$,
$\forall t\in (\delta, t_1)$.
Grownwall's inequality then yields $y(t_1)\le x_1e^{c_2t_1 }<c_3$, which is a contradiction.
Hence $t_1=T$ and $y(t)\le c_3$ on $[\delta, T]$. In this case equations (\ref{control}) and  (\ref{eqy2}) coincide.
\end{proof}

Combing Propositions \ref{localex4eqy} and \ref{globalex4eqy}, we arrive at

\begin{theorem}\label{odeex}
Under Assumption  \ref{h13},  the equation (\ref{Lambdaeq}) admits a solution
$\Lambda(\cdot)\in C[0,T]\cap C^1[0,T)$ satisfying $\Lambda(t)>0$
$\forall t\in [0, T)$.
\end{theorem}

%

\subsection{Inequality (\ref{2ndcond})}

We now provide conditions on the model primitives under which the inequality (\ref{2ndcond}) holds. We assume that the ODE  (\ref{Lambdaeq}) admits a positive solution
$\Lambda(\cdot)\in C[0,T]\cap C^1[0,T)$, and
introduce the following additional assumption.

\begin{assump}\label{condonu}
   For a.e. $t\in [0,T)$, $w(t,\cdot)$ is either convex or inverse S-shaped.
  \end{assump}

A convex weighting function captures risk aversion in terms of exaggerating the small probability of very ``bad" events while downplaying very ``good" events; see Yaari (1987).
On the other hand, as discussed in Introduction, an inverse S-shaped probability weighting function is more interesting as it reflects
the tendency of inflating the small probabilities of {\it both} tails which are consistent
with the conclusions of many experimental and empirical works.

\begin{lemma}
Under  Assumptions \ref{h13} and \ref{condonu}, we have 
\begin{equation}\label{condintbypt}
\lim_{|y|\rightarrow +\infty}\frac{w'_p(t, N (y))N'(y)}{I'(e^{\sqrt{\Lambda(t)}y-c})}=0\;\;\forall c\in\R, \;\forall t\in[0,T).
\end{equation}
\end{lemma}
\begin{proof} 
It follows from (\ref{k1}) that
 $$0\le \frac{w'_p(t, N (y))}{-I'(e^{\sqrt{\Lambda(t)}y-c})}\le K w'_p(t, N (y))( e^{\alpha \sqrt{\Lambda(t)}y-\alpha c}+1).$$
 However
 $$\int_{-\infty}^{+\infty} w'_p(t, N (y))( e^{\alpha \sqrt{\Lambda(t)}y-\alpha c}+1) N'(y)dy=e^{-\alpha c} h(t,\alpha\sqrt{\Lambda(t)})+1<+\infty,$$
 implying $$\lim_{|y|\rightarrow +\infty}w'_p(t, N (y))( e^{\alpha \sqrt{\Lambda(t)}y-\alpha c}+1) N'(y)=0.$$
 This completes the proof.
\end{proof}

%
%
%

We now analyze the integral in (\ref{2ndcond}) by decomposing it into $M_0+M$, where
$$M_0=\int_{-\infty}^{+\infty} w_p'\left(t,N \left(\frac{c -g(x)}{\sqrt{\Lambda(t)}}\right)\right)
  N'\left(\frac{c -g(x)}{\sqrt{\Lambda(t)}}\right)
      \frac{c -g(x)}{\Lambda(t)}g'(x)^2 du(x),
$$ and
\begin{eqnarray*}
  M&=& \int_{-\infty}^{+\infty} w_p'\left(t,N \left(\frac{c -g(x)}{\sqrt{\Lambda(t)}}\right)\right)
  N'\left(\frac{c -g(x)}{\sqrt{\Lambda(t)}}\right)g''(x)du(x)\\
   &=&\int_{ -\infty}^{ +\infty} w_p'\left(t,N \left(\frac{c -g(x)}{\sqrt{\Lambda(t)}}\right)\right)
  N'\left(\frac{c -g(x)}{\sqrt{\Lambda(t)}}\right)
      u'( x)dg'(x).
\end{eqnarray*}
Applying integration by parts and noting that $N^\dpm(x)=-xN'(x)$,  we can further decompose $M$ into $-(M_1+M_2+M_3)$ where
\begin{eqnarray*}
  M_1&=&\int_{ -\infty}^{ +\infty} g'(x)w_p^\dpm\left(t,N \left(\frac{c -g(x)}{\sqrt{\Lambda(t)}}\right)\right)
  \left(N'\left(\frac{c -g(x)}{\sqrt{\Lambda(t)}}\right)\right)^2 u'(x)\frac{-g'(x)}{\sqrt{\Lambda(t)}}dx\\
  M_2&=&\int_{ -\infty}^{ +\infty} g'(x)w_p'\left(t,N \left(\frac{c -g(x)}{\sqrt{\Lambda(t)}}\right)\right)
  N^\dpm\left(\frac{c -g(x)}{\sqrt{\Lambda(t)}}\right)u'(x)\frac{-g'(x)}{\sqrt{\Lambda(t)}}dx\\
     &=&-\int_{ -\infty}^{ +\infty} g'(x)w_p'\left(t,N \left(\frac{c -g(x)}{\sqrt{\Lambda(t)}}\right)\right)
  N'\left(\frac{c -g(x)}{\sqrt{\Lambda(t)}}\right)u'(x)\frac{c -g(x)}{\sqrt{\Lambda(t)}}\frac{ -g'(x)}{\sqrt{\Lambda(t)}}dx\\
     &=&\int_{ -\infty}^{ +\infty} g'(x)^2w_p'\left(t,N \left(\frac{c -g(x)}{\sqrt{\Lambda(t)}}\right)\right)
  N'\left(\frac{c -g(x)}{\sqrt{\Lambda(t)}}\right)u'(x)\frac{c -g(x)}{\sqrt{\Lambda(t)}}dx\\
     &\equiv&M_0,\\
  M_3&=&\int_{ -\infty}^{ +\infty} g'(x)w_p'\left(t,N \left(\frac{c -g(x)}{\sqrt{\Lambda(t)}}\right)\right)
  N'\left(\frac{c -g(x)}{\sqrt{\Lambda(t)}}\right)
      u^\dpm( x)dx,
\end{eqnarray*}
{\it assuming} that for any $c\in \R$,
\begin{equation}\label{midstep}
w_p'\left(t,N \left(\frac{c -g(x)}{\sqrt{\Lambda(t)}}\right)\right)
  N'\left(\frac{c -g(x)}{\sqrt{\Lambda(t)}}\right)u'(x)g'(x)=0 \mbox{ when } |x|\rightarrow +\infty.
  \end{equation}
Since $g'(x)=-\frac{u^\dpm(x)}{u'(x)}$, (\ref{midstep}) can be written as
$$
w_p'\left(t,N \left(\frac{c -g(x)}{\sqrt{\Lambda(t)}}\right)\right)
  N'\left(\frac{c -g(x)}{\sqrt{\Lambda(t)}}\right)u^\dpm(x)=0\mbox{ when } |x|\rightarrow  \infty,$$
which is  equivalent to (\ref{condintbypt}).

Now, denoting $y=\frac{c-g(x)}{\sqrt{\Lambda(t)}}$, we have
\begin{eqnarray*}
 && M_0+M=-(M_1+M_3)\\
 &=&-\int_{-\infty}^{+\infty}g'(x)N'(y)
   \left(w_p^\dpm\left(t,N(y)\right) N'(y)
      u'(x)\da{y}{x}+w_p'\left(t,N(y))\right) u^\dpm(x) \right)dx\\
   &=&-\int_{-\infty}^{+\infty} g'(x)N'(y)d\left[w_p'\left(t,N(y)\right) u'(x)\right]\\
   &=&-\int^{-\infty}_{+\infty} g'(f(c -\sqrt{\Lambda(t)}y))N'(y)
       d\left[w_p'\left(t,N(y)\right) u'(f(c -\sqrt{\Lambda(t)}y))\right]\\
   &=&\int_{-\infty}^{+\infty} \left[\frac{N'(y)}{f'(c -\sqrt{\Lambda(t)}y)}\right]
       d\left[w_p'\left(t,N(y)\right) u'(f(c -\sqrt{\Lambda(t)}y))\right].
\end{eqnarray*}
Hence, inequality (\ref{2ndcond}) is equivalent to
\begin{equation}\label{2odcon2}
\int_{-\infty}^{+\infty} \frac{N'(y)}{f'(c -\sqrt{\Lambda(t)}y)}
       d\left[w_p'\left(t,N(y)\right) u'(f(c -\sqrt{\Lambda(t)}y))\right]\ge 0
\end{equation}
for any $c\in \R$,  a.e. $t\in [0, T)$.

An obvious  sufficient condition for (\ref{2odcon2}) is that
$$w_p'\left(t,N(y)\right) u'(f(c -\sqrt{\Lambda(t)}y))\equiv w_p'\left(t,N(y)\right) e^{\sqrt{\Lambda(t)}y-c}$$
 is increasing in $y$ for any $c\in \R$,
which holds automatically if $w(t,\cdot)$ is convex. 

\begin{theorem}\label{2ndconsuff}
Under Assumptions \ref{h13} and \ref{condonu}, the inequality (\ref{2ndcond}) holds.
\end{theorem}
\begin{proof}
We have shown above that (\ref{2ndcond}) holds at $t$ when $w(t,\cdot)$ is convex. Let us now focus on the case when $w(t,\cdot)$ is inverse S-shaped.

First of all, it follows from Theorem \ref{odeex} that the ODE  (\ref{Lambdaeq}) admits a positive solution
$\Lambda(\cdot)\in C[0,T]\cap C^1[0,T)$.

Fix $t\in[0,T)$. Noting  $f'(x)=-\tilde\kappa I'(\tilde\kappa e^{-x})e^{-x}$
where $\tilde \kappa=\kappa e^{-\frac{1}{2}\Lambda(0)}>0$, 
we can rewrite condition (\ref{2odcon2}) as
\begin{eqnarray*}
0&\le& \int^{+\infty}_{-\infty}\frac{N'(y)}{f'(c -\sqrt{\Lambda(t)}y)}
       d\left[w_p'\left(t,N(y)\right) u'(f(c -\sqrt{\Lambda(t)}y))\right]\\
&=&\int^{+\infty}_{-\infty}\frac{N'(y)}{-\tilde\kappa I'(\tilde\kappa e^{\sqrt{\Lambda(t)}y-c})e^{\sqrt{\Lambda(t)}y-c}} d\left[w_p'\left(t,N(y)\right) e^{\sqrt{\Lambda(t)}y-c}\right]\\
&=&\int^{+\infty}_{-\infty}\frac{N'(y)}{-\tilde\kappa I'(\tilde\kappa e^{\sqrt{\Lambda(t)}y-c})}\left[w_p'\left(t,N(y)\right)
\sqrt{\Lambda(t)}+\frac{\partial w_p'(t,N(y))}{\partial y}\right] dy \\
&=&\sqrt{\Lambda(t)}M_4+M_5,
\end{eqnarray*}
where $M_4=\int^{+\infty}_{-\infty}\frac{w_p'\left(t,N(y)\right)N'(y)}{-\tilde\kappa I'(\tilde\kappa e^{\sqrt{\Lambda(t)}y-c})}dy>0$, and
\begin{eqnarray*}
M_5&=&\int^{+\infty}_{-\infty}\frac{\frac{\partial w_p'(t,N(y))}{\partial y}}{-\tilde\kappa I'(\tilde\kappa e^{\sqrt{\Lambda(t)}y-c})}N'(y) dy\\
&=&\E[F(\xi)G(\xi)]
\end{eqnarray*}
with $F(y)=\frac{\partial w_p'(t,N(y))}{\partial y}$ and $G(y)=\frac{1}{-\tilde\kappa I'(\tilde\kappa e^{\sqrt{\Lambda(t)}y-c})}>0$.

Since $w(t,\cdot)$ is inverse S-shaped, there exists $q\in (0,1)$ such that $w_p^\dpm(t,p)\leq 0$ on $(0, q)$, $w_p^\dpm(t,p)\geq 0$ on
$(q, 1)$, and $w_p^\dpm(t,q)=0$.
Furthermore, for any $b>0$,
\begin{eqnarray*}
&&\int^{b}_{0} \frac{\partial w_p'(t,N(y))}{\partial y}N'(y)dy\\
&=&\int^{b}_{0}N'(y)d[w_p'(t,N(y))]\\
&=&(w_p'(t,N(y))N'(y))|_{y=0}^{y=b}-\int^{b}_{a}w_p'(t,N(y))dN'(y)\\
&=&w_p'(t,N(b))N'(b)-\frac{1}{\sqrt{2\pi}}w_p'(t,\frac{1}{2})+\int^{b}_{0}w_p'(t,N(y))y dN(y).
\end{eqnarray*}
Since $w_p^\dpm(t,p)\geq 0$ for sufficiently large $p$, $\int^{b}_{0} \frac{\partial w_p'(t,N(y))}{\partial y}N'(y)dy$ is increasing in $b$ for sufficiently large $b$. Hence
$\lim_{b\uparrow +\infty} \int^{b}_{0} \frac{\partial w_p'(t,N(y))}{\partial y}N'(y)dy$ exists.
However, $+\infty>h'_x(t,0)\equiv \E[w_p'(t,N(\xi))\xi]$,  we conclude that $\lim_{b\uparrow +\infty} \int^{b}_{0}w_t'( N(y))y dN(y)$ exists and is finite.
This implies that $a:=\lim_{b\uparrow +\infty}w_p'(t,N(b))N'(b)$ exists. It then follows from the fact that $+\infty>h(t,0)\equiv \E[w_p'(t,N(\xi))]$ that $a=0$, i.e.,
$$\int^{+\infty}_{0} \frac{\partial w_p'(t,N(y))}{\partial y}N'(y)dy=-\frac{1}{\sqrt{2\pi}}w_p'\left(t,\frac{1}{2}\right)+\int^{+\infty}_{0}w_p'(t,N(y))y dN(y).
$$
Similarly, we can show that
$$ \int^{0}_{-\infty} \frac{\partial w_p'(t,N(y))}{\partial y}N'(y)dy=\frac{1}{\sqrt{2\pi}}w_p'\left(t,\frac{1}{2}\right)+\int^{0}_{-\infty} w_p'(t,N(y))y dN(y).
$$
As a result,
$$\E[F(\xi)]=\int^{+\infty}_{-\infty} \frac{\partial w_p'(t,N(y))}{\partial y}N'(y)dy=\int_{-\infty}^{+\infty}w_p'(t,N(y))y dN(y)=h'_x(t,0)\ge 0.$$

Denote $c=N^{-1}(q)$. Then $F(\xi)$ is negative on $(-\infty, c)$, positive on $(c, +\infty)$, and $F(c)=0$. Since $u^\dpm(x)$ is increasing in $x$, so is $I'(x)=1/u^\dpm(I(x))$; hence
$G(\cdot)$ is increasing. Then we have $F(\xi)(G(\xi)-G(c))\ge 0$. Consequently,
$$0\le\E[F(\xi)(G(\xi)-G(c))]=\E[F(\xi)G(\xi)]-G(c)\E[F(\xi)]=\E[F(\xi)G(\xi)]-G(c)h'_x(t,0),$$
which implies $$\E[F(\xi)G(\xi)]\ge G(c)h'_x(t,0)\ge 0.$$ The proof is complete.
\end{proof}

\section{An Example}

In this section, we give a concrete example to demonstrate our results.

Define $$F(x):=N(x/\sqrt{2}), \mbox{ and } w(t,p):=F(N^{-1}(p)),$$
where  $N$ is the probability distribution function of a standard normal.
So in this example our weighting function is time-invariant, and we will drop $t$ and
write $w(p)=w(t,p)$ throughout this section. 

It is clear that $w(0)=0, w(1)=1$, $w'(p)=\frac{F'(N^{-1}(p))}{N'(N^{-1}(p))}>0$. So $w$ is a probability weighting function satisfying Assumption \ref{basicassump}-(i). Moreover, we have
\begin{eqnarray*}
w(N(x))&=&F(x)=N(x/\sqrt{2}),\\ 
w'(N(x))&=&\frac{F'(x)}{N'(x)}=\frac{\frac{1}{\sqrt{2}}N'(x/\sqrt{2})}{N'(x)}
=\frac{1}{\sqrt{2}}e^{x^2/4},\\
h(t,x)&\equiv &\E[w'(N(\xi))e^{x\xi}]=\frac{1}{\sqrt{2}} \int_{\R} \frac{1}{\sqrt{2\pi}}e^{xz}e^{-z^2/4}dz=e^{x^2}.
\end{eqnarray*}

Now, take any $m\in (-1, -1/2)$. Then
\begin{eqnarray*}
\lim_{p\downarrow 0}\frac{w'(p)}{p^m}&=&
\lim_{x\rightarrow -\infty}\frac{w'(N(x))}{N(x)^m}\\
&=&\lim_{x\rightarrow -\infty}\frac{1}{\sqrt{2}}\frac{e^{x^2/4}}{N(x)^m}\\
&=&\frac{1}{\sqrt{2}}\left(\lim_{x\rightarrow -\infty}\frac{e^{x^2/(4m)}}{N(x)}\right)^m\\
&=&\frac{1}{\sqrt{2}}\left(\lim_{x\rightarrow -\infty}\frac{xe^{x^2/(4m)}/2m}{N'(x)}\right)^m\\
&=&\frac{1}{\sqrt{2}}\left(\frac{1}{2m} \lim_{x\rightarrow -\infty} xe^{x^2(\frac{1}{4m}+\frac{1}{2})}\right)^m\\
&=&0.
\end{eqnarray*}
Also, we have $\lim_{p\uparrow 0}\frac{w'(p)}{(1-p)^m}=\lim_{q\downarrow 0}\frac{w'(1-q)}{q^m}=\lim_{q\downarrow 0}\frac{w'(q)}{q^m}=0.$ So Assumption \ref{uandw}-(iv) holds.

Since $h(t,x)=e^{x^2}$, it is straightforward to verify that Assumption \ref{h13} is satisfied.
Moreover, $w'(p)=\frac{1}{2}e^{(N^{-1}(p))^2/4}$; so $w$ is concave on $[0,1/2]$ and
convex on $[1/2,1]$. It is therefore inverse S-shaped, satisfying Assumption \ref{condonu}. We have now checked the validity of all the assumptions on $w$ in this paper.

With the explicit form of $h$,  the ODE (\ref{Lambdaeq}) becomes  
$$\Lambda'(t)=-\frac{1}{4}|\theta(t)|^2,\;\; t\in[0,T), \qquad \Lambda(T)=0,$$
whose solution is $\Lambda(t)=\frac{1}{4}\int_t^T|\theta(s)|^2ds$. Thus,
$$\lambda(t)=\sqrt{-\Lambda'(t)/|\theta(t)|^2}=1/2.$$

As discussed at the end of Section 4.1, $\lambda(t)=1/2$ means that
the RDU agent in this example acts like an EUT agent with the  market price of risk
is revised to half of the original one. In other words, because of the probability weighting and the consistent planning (recall that consistent planning amounts to
a heavily constrained optimization), the agent will lose half of the risk premium or Sharpe ratio.

To complete this example, we present the equilibrium portfolio $\pi^*$ when
 the exponential utility function is $u(x)=1-e^{-\alpha x}$, $x\in\R$, for some $\alpha>0$, which clearly satisfies all the assumptions on the utility function in this paper.

 In this case, $u'(x)=\alpha e^{-\alpha x}$ and hence $I(x)=\frac{\ln \alpha-\ln x}{\alpha}$, $x>0$. It is easy to check that there exists  $\kappa>0$ satisfying the budget constraint (\ref{bc}). Hence the desired terminal
wealth is
$$X^*(T)=I(\kappa\bar \rho(T))=c+\frac{1}{\alpha}\int_0^T\lambda(t)\theta(t)^\top  d\tilde  W(t),$$
where $\tilde W(t)=W(t)+\int_0^t\theta(s)ds$ is a standard Brownian motion under the risk-neutral measure, $\Q$, of the original market, and $c$ is a constant dependent of the market parameters.

By the pricing theory, the replicating wealth process $X^*(\cdot)$ of $X^*(T)$ is a $\Q$-martingale (recall that the risk-free rate has been assumed to be 0); hence
$$ X^*(t)=c+\frac{1}{\alpha}\int_0^t\lambda(s)\theta(s)^\top d\tilde  W(s),\;\;t\in[0,T].$$
Matching the above with the wealth equation $dX^*(t)=\pi^*(t)^\top \sigma(t) d\tilde  W(t)$, we obtain the equilibrium portfolio
$$\pi^*(t)=\frac{1}{\alpha}\lambda(t)(\sigma(t)^\top)^{-1}\theta(t).$$
Recall that the optimal portfolio of an EUT agent with the same exponential utility is
$\pi^{EUT}(t)=\frac{1}{\alpha}(\sigma(t)^\top)^{-1}\theta(t).$ Hence the risky investment of the RDU agent at $t$ is that of the EUT agent multiplied by $\lambda(t)$. Note that this result does not depend on the spcific form of the weighting function so long as $\lambda(\cdot)$ exists. In the special case when $w(t,p)=F(N^{-1}(p))$,
$\lambda(t)\equiv 1/2$; so the risk exposure is reduced by
half.

\section{Reduction in Risk Premium}

We have seen that the sophisticated RDU agent behaves as if the risk premium is factored by $\lambda(\cdot)$. In the example with the specific probability weighting function presented in Section 5, $\lambda(t)=1/2$; so the risk premium is reduced and the agent acts more cautiously than her EUT counterpart. In this section, we answer the general question
of when there is a reduction in risk premium or, equivalently, when 
$\lambda(t)< 1$. 

We assume $\Lambda(\cdot)$ exists. 
It follows from  (\ref{Lambdaeq}) that
$$\lambda(t)=\left| \frac{\sqrt{\Lambda(t)}h(t, \sqrt{\Lambda(t)}\,)}{h'_x(t, \sqrt{\Lambda(t)}\,)}\right|>0,\;\;t\in[0,T).$$


For any $x\in \R^+$, define a probability measure $\Q^x$
by $\frac{d\Q^x}{d\p}=e^{-\frac{x^2}{2}+x\xi}$, under which $\xi-x\sim N(0,1)$.
So
\begin{eqnarray*}
h(t,x)&=&e^{\frac{x^2}{2}}\E[w'_p(t,N(\xi))e^{-\frac{x^2}{2}+x\xi}]\\
&=&e^{\frac{x^2}{2}}\E^{\Q^x}[w'_p(t,N(\xi-x+x))]\\
&=&e^{\frac{x^2}{2}}\E[w'_p(t,N(\xi+x))]\\
&=&e^{\frac{x^2}{2}}H(t,x),
\end{eqnarray*}
where $H(t,x):=\E[w'_p(t,N(\xi+x))]\geq0$. The following characterizes the condition $\lambda(t)< 1$.

\begin{theorem}\label{lambda<1}
Assume that $\Lambda(\cdot)$ exists. Then, for any $t\in[0,T)$,
$\lambda(t)< 1$ if and only if $H'_x(t,\sqrt{\Lambda(t)})> 0$.
\end{theorem}
\begin{proof}
We have
\[
\frac{x h(t,x)}{h'_x(t,x)}=\frac{x e^{\frac{x^2}{2}} H(t,x)}{e^{\frac{x^2}{2}}H'_x(t,x)+xe^{\frac{x^2}{2}}H(t,x)}=\frac{x H(t,x)}{H'_x(t,x)+xH(t,x)}.
\]
Noting $\lambda(t)>0\;\forall t\in[0,T)$, we conclude that, for any $t\in[0,T)$,
$\lambda(t)=\frac{x h(t,x)}{h'_x(t,x)}|_{x=\sqrt{\Lambda(t)}}< 1$ if and only if $H'_x(t,x)|_{x=\sqrt{\Lambda(t)}}\equiv H'_x(t,\sqrt{\Lambda(t)})> 0$.
\end{proof}


\begin{corollary}\label{incre-H}
For any $t\in[0,T)$, $\lambda(t)<1$ if $H(t,x)$ is strictly increasing in $x\in \R^+$.
In particular, if $w(t,p)$ is strictly convex in $p$, then $\lambda(t)<1$.
\end{corollary}
\begin{proof} This is obvious.
\end{proof}

In the example in Section 5, $H(t,x)= e^{-x^2/2}h(t,x)=e^{x^2/2},
$
which is strictly increasing in $x\in \R^+$; hence Corollary \ref{incre-H} applies.
On the other hand, as explained earlier, a strict convexity of $w(t,p)$ in $p$ underlines strict overweighting of the left tail (i.e. the bad events) and strict underweighting of the right tail (i.e. the good events); so it enhances the level of risk aversion leading to a smaller risk premium and less risky exposure.

Even if $w$ is more general including being inverse S-shaped, it is still possible that $\lambda(t)<1$, as demonstrated by the example in Section 5.
In this case, the equivalent condition $H'_x(t,\sqrt{\Lambda(t)})> 0$ can be used to uncover the balance between the risk-averse component and the risk-seeking component 
in probability weighting needed to render an overall reduction in risk premium.  We leave a detailed study along this direction to interested readers.

\section{Concluding Remarks}

A continuous-time RDU portfolio selection problem is inherently time inconsistent.
A sophisticated agent, realizing that in the future she herself might disagree to her current planning, resorts to consistent investment by implementing intra-personal equilibrium strategies from which she will have no incentive to deviate at any point in time.
We have solved the open problem  of deducing such equilibrium strategies, by developing an approach that to our best knowledge is new to the literature.
The main technical thrust of our approach is to express the  first-order derivative
of the small deviation from an equilibrium as a quadratic function of the deviating amount. The definition of the equilibrium requires the quadratic function to have a constant sign whatever the amount might be. This leads to an equality (which in turn leads to an ODE) and an inequality, the two constituting   the main sufficient conditions for deriving explicitly the final wealth profile and, hence, the resulting strategy courtesy of the market completeness.

With an intra-personal equilibrium strategy, the agent, at any given time, in effect solves a
{\it constrained}  dynamically optimal RDU model in which the constraint is to honor 
{\it all} her future strategies. One may think this would lead to an extremely complicated terminal wealth profile. Our result, however, shows that the terminal wealth is surprisingly simple -- it resembles that of an optimal Merton portfolio, except that the investment opportunity set needs to be modified properly. In other words, the RDU agent behaves as if she was an EUT agent, only that she is in a fictions market where she blends her probability weighting function into the market price of risk.
This observation may in turn shed lights on finding intertemporal market equilibria for markets
where all the agents are EUT and/or RDU consistent planners.

We derived our equilibrium strategy based on an Ansatz; as such, our results do not rule out the possibility of having other equilibria beyond our Ansatz.
In general, uniqueness of intra-personal equilibrium for time-inconsistent problems
remains a very challenging research question.\footnote{To our best knowledge, Hu et al (2017) is the only paper that addresses the uniqueness in continuous time.} It is particularly the case for the RDU model, or so we believe.

\newpage

\noindent{\Large\bf Appendices}

\appendix

\section{Existence of Solutions to a Class of ODEs}

In this appendix we present some general existing results on a class of ODEs, taken from Agarwal and  O'Regan (2004).
Consider the following ODE
\begin{equation}\label{ivpsg}
\left\{\begin{array}{l}
   y'(t)=f(t,y(t)), \quad t\in (0, T],\\
   y(0)=0\in \R,
   \end{array}\right.
 \end{equation}
 where $f$ may not be defined at $t=0$.

 Denote by $AC[0, s]$ the set of absolutely continuous functions on $[0,s]$ where $s>0$. The
 following assumption is introduced in Agarwal and  O'Regan (2004).

\begin{assump}\label{ivpas1}
There exists $t_0\in(0,T]$ such that
\begin{enumerate}
 \item There is $\beta\in AC[0, t_0]\cap C^{1}(0,t_0]$ with $\beta(0)\ge 0$ such that
$$\beta'(t)\ge f(t, \beta(t)), \;t\in (0, t_0].$$
\item There is $\alpha\in AC[0, t_0]\cap C^{1}(0,t_0]$ with $\alpha(t)\le \beta(t)$  $\forall t\in [0, t_0]$ and
$\alpha(0)\le 0$ such that
$$\alpha'(t)\le f(t, \alpha(t)), \;t\in (0, t_0]. $$
\item
The function $$f^{*}(t,y)=\left\{\begin{array}{lll}
f(t, \beta(t))+g(\beta(t)-y)&& y\ge \beta(t),\\
f(t,y)&& \alpha(t)<y<\beta(t),\\
f(t, \alpha(t))+g(\alpha(t)-y)&& y\le \alpha(t),
\end{array}\right.$$
is, in the region $(0, t_0]\times \R$,  continuous in $y$ for any $t$ and measurable in $t$ for any $y$,
where $g(x)=x\id_{|x|\le 1}+\rm{sign}(x)\id_{|x|>1}$ is the radial retraction.\footnote{In \cite{cdf},
$f^{*}$ is assumed to be (jointly) continuous; but from the proofs therein, we can easily weaken it to this current
version.}
\end{enumerate}
\end{assump}

The following two propositions, both drawn from Agarwal and  O'Regan (2004), concern the local and global existence of the ODE (\ref{ivpsg})
respectively.

\begin{proposition}[Agarwal and  O'Regan 2004, Theorem 3.1]\label{localex}
Under Assumption \ref{ivpas1}, the ODE (\ref{ivpsg}) admits a solution $y\in AC[0, t_0]$ satisfying
$\alpha(t)\le y(t)\le \beta(t)$ for $t\in [0, t_0]$.
\end{proposition}

\begin{proposition}[Agarwal and  O'Regan 2004, Theorem 1.4]\label{globalex}
Given $t_0\in(0,T]$, assume that $f(t,y)$ is continuous in $y$ for any $t\in [t_{0}, T]$,
measurable in $t$ for any $y\in \R$,
and there exists $g(\cdot)\in L^{1}[t_{0},T]$ such that
$|f(t, y)|\le g(t)$ $\forall (t,y)\in [t_{0},T]\times \R$.
Then the equation
\begin{equation}\label{ivpsg0}
\left\{\begin{array}{l}
   y'(t)=f(t,y(t)), \quad t\in (t_{0}, T],\\
   y(t_{0})=a\in \R,
   \end{array}\right.
 \end{equation}
admits a solution $y\in AC[t_0,T]$.
\end{proposition}


\section{Tversky--Kahnamen's Probability Weighting Functions}\label{KT-check}

In this appendix we verify that a class of time--varying  Tversky--Kahnamen (TK) probability weighting functions satisfy all the technical assumptions required in the paper.

First of all, the original TK weighting function, introduced in Tversky and Kahnamen (1992), is \begin{equation}\label{TKW}
w_{TK}(p;\delta):=\frac{p^\delta}{[p^\delta+(1-p)^\delta]^{1/\delta}},\;\;p\in[0,1],
\end{equation}
 where
$\delta\in (0,1]$ is a parameter. This is an inverse S-shaped function, with $w_{TK}'(p;\delta)>1$ when $p$ is close to both 0 and 1. Moreover, a smaller
$\delta$ implies a stronger degree of probability weighting. When $\delta=1$, there is no probability weighting.

We now vary the parameter $\delta$ over time to generate a family of time-dependent TK
functions.
Given a measurable function $\delta: [0,T]\mapsto (0,1]$ with $0<\delta(t)<1$ for $t\in[0,T)$ and $\delta(T)=1$, define
\begin{equation}\label{wtp}
\tilde w(t,p):=w_{TK}(p;\delta(t)),\;\;(t,p)\in [0,T]\times [0,1].
\end{equation}
The purpose of this appendix is to show that $\tilde w$ satisfies all the assumptions in the paper, under proper conditions on $\delta(\cdot)$.

Clearly, $\tilde w$ satisfies Assumption \ref{basicassump}-(i).

\begin{proposition}\label{uandwiv}
$\tilde w$ satisfies  Assumption \ref{uandw}-(iv).
\end{proposition}

\begin{proof} As $t\in[0,T)$ is fixed in Assumption \ref{uandw}-(iv), from the construction of $\tilde w$ it suffices to prove the conclusion for $w_{TK}(\cdot;\delta)$
with $\delta\in (0,1)$ fixed.

Denote a function
\begin{equation}\label{alphapd}
\alpha(p;\delta):=[p^\delta+(1-p)^\delta]^{-\frac{1}{\delta}},\;\;p\in[0,1],
\end{equation}
with parameter $\delta\in (0,1)$.  It is easy to see that
$1\leq p^\delta+(1-p)^\delta< 2$; hence
\begin{equation}\label{alphabd}
2^{-\frac{1}{\delta}}< \alpha(p;\delta)\leq 1.
\end{equation}
Moreover,
$\alpha(p;\delta)=\alpha(1-p;\delta)$, $w_{TK}(p;\delta)=p^\delta \alpha(p;\delta)$. Hence
\begin{eqnarray*}
 \alpha'(p;\delta)&=&-\frac{1}{\delta}[p^\delta+(1-p)^\delta]^{-\frac{1}{\delta}-1}\delta[ p^{\delta-1}-(1-p)^{\delta-1}]\\
&=&-\alpha(p;\delta)^{1+\delta}[p^{\delta-1}-(1-p)^{\delta-1}],\\
w'_{KT}(p;\delta)&=&\delta p^{\delta-1}\alpha(p;\delta)+p^\delta \alpha'(p;\delta)\\
&=&\delta \alpha(p;\delta)p^{\delta-1}
  +\alpha(p;\delta)^{1+\delta} p^\delta (1-p)^{\delta-1}-\alpha(p;\delta)^{1+\delta} p^{2\delta-1}\\
  &<& p^{\delta-1}+(1-p)^{\delta-1},
\end{eqnarray*}
where the last inequality is by (\ref{alphabd}). This completes the proof.
\end{proof}

Before we move to the next assumption,
for any probability weighting function $w: [0,1]\mapsto [0,1]$, define
(with a slight abuse of notation)
$h(x;w):=\E[w'(N(\xi))e^{x\xi}]$, $x\in\R$.
It then follows from Lemma \ref{hregularity}-(ii) that, for any {\it odd} number $n\geq1$, 
\begin{eqnarray}
 h^{(n)}(0;w)&=&\E[w'(N(\xi);\delta)\xi^{n}]\\
&=& \int_{-\infty}^0w'(N(y))y^{n}N'(y)dy+\int_0^{+\infty}w'(N(y))y^{n}N'(y)dy\nonumber\\
&=&\int_{+\infty}^0 w'(N(-y))(-y)^{n}N'(-y)d(-y)+\int_0^{+\infty}w'(N(y))y^{n}N'(y)dy\nonumber\\
&=&-\int_0^{+\infty} w'(1-N(y))y^{n}N'(y)dy+\int_0^{+\infty}w'(N(y))y^{n}N'(y)dy\nonumber\\
&=&\int_0^{+\infty}[w'(N(y))- w'(1-N(y))]y^{n}N'(y)dy,\label{hderiv}
\end{eqnarray}
and
\begin{eqnarray}
 h^\dpm(0;w)&=&\int_{-\infty}^0w'(N(y))y^{2}N'(y)dy+\int_0^{+\infty}w'(N(y))y^{2}N'(y)dy\nonumber\\
&=&\int_{+\infty}^0 w'(N(-y))(-y)^{2}N'(-y)d(-y)+\int_0^{+\infty}w'(N(y))y^{2}N'(y)dy\nonumber\\
&=&\int_0^{+\infty} w'(1-N(y))y^{2}N'(y)dy+\int_0^{+\infty}w'(N(y))y^{2}N'(y)dy\nonumber\\
&=&\int_0^{+\infty}[w'(N(y))+ w'(1-N(y))]y^{2}N'(y)dy.\label{hderiv2}
\end{eqnarray}

\begin{proposition}
$\tilde w$ satisfies Assumption \ref{h13}-(i).
\end{proposition}

\begin{proof} With a slight abuse of notation, define
$h(x;\delta):=\E[w'_{KT}(N(\xi);\delta)e^{x\xi}]$, $x\in\R$, where $\delta\in(0,1]$.
Again, as $t\in[0,T]$ is fixed in Assumption \ref{h13}-(i), we can drop $t$ and need only to show that $h'(0;\delta)\geq0$ and $h'''(0;\delta)\geq0$.

Noting that $w'_{TK}(1-p;\delta) =\delta (1-p)^{\delta-1}\alpha(p;\delta)-(1-p)^\delta \alpha'(p;\delta),$
 we have
\begin{equation}\label{wkt}
\begin{array}{rcl}
 w'_{TK}(p;\delta)-w'_{TK}(1-p;\delta)&=&\delta p^{\delta-1}\alpha(p;\delta)+p^\delta \alpha'(p;\delta)-
 \left[\delta (1-p)^{\delta-1}\alpha(p;\delta)-(1-p)^\delta \alpha'(p;\delta)\right]\\
&=&\delta\alpha(p;\delta)[p^{\delta-1}-(1-p)^{\delta-1}]
+\alpha'(p;\delta)[p^\delta+(1-p)^\delta]\\
&=&\delta \alpha(p;\delta)[p^{\delta-1}-(1-p)^{\delta-1}]+\alpha'(p;\delta) \alpha(p;\delta)^{-\delta}\\
&=&\delta \alpha(p;\delta)[p^{\delta-1}-(1-p)^{\delta-1}]
+\alpha(p;\delta)[(1-p)^{\delta-1}-p^{\delta-1}]\\
&=&(\delta-1)\alpha(p;\delta) [p^{\delta-1}-(1-p)^{\delta-1}]>0\;\;\forall p\in (1/2, 1].
\end{array}
\end{equation}
Now Assumption \ref{h13}-(i) holds by virtue of (\ref{hderiv})
where we take $w=w_{PK}(\cdot;\delta)$.
\end{proof}

\begin{proposition}
If  $\inf_{t\in [0,T]} \delta(t)>0$, then $\tilde w$
satisfies Assumption \ref{h13}-(iv).
\end{proposition}
\begin{proof} We have
\begin{eqnarray*}
w'_{TK}(p;\delta)+w'_{TK}(1-p;\delta)&=&\delta p^{\delta-1}\alpha(p;\delta)+p^\delta \alpha'(p;\delta)+
 \left[\delta (1-p)^{\delta-1}\alpha(p;\delta)-(1-p)^\delta \alpha'(p;\delta)\right]\\
&=&\delta\alpha(p;\delta)[p^{\delta-1}+(1-p)^{\delta-1}]
+\alpha'(p;\delta)[p^\delta-(1-p)^\delta]\\
&>&\delta\alpha(p;\delta)[p^{\delta-1}+(1-p)^{\delta-1}]\\
&> &2^{-\frac{1}{\delta}}\delta \;\;\forall p\in (1/2, 1].
\end{eqnarray*}
It hence follows from
(\ref{hderiv2}) that
$$h^\dpm_x(t, 0)>\int_0^{+\infty}2^{-\frac{1}{\delta(t)}}\delta(t)y^2dN(y)
=2^{-\frac{1}{\delta(t)}}\delta(t)/2\ge \inf_{t\in [0,T]}[2^{-\frac{1}{\delta(t)}}\delta(t)]/2>0, $$
owing to the condition that $\inf_{t\in [0,T]} \delta(t)>0$.
\end{proof}

For Assumption \ref{h13}-(iii), we need the following lemma.
\begin{lemma}\label{suff-(iii)}
For any $\epsilon\in (0,1)$, we have
$\E[N(\xi)^{-\epsilon} e^{x\xi}]<+\infty\;\; \forall x\in \R$.
\end{lemma}
\begin{proof}
The statement is symmetric for $x\le 0$ and $x\ge 0$; hence it suffices to prove for the case when $x<0$ (the case of $x=0$ is trivial).

Since $\E[N(\xi)^{-\epsilon} e^{x\xi}\id_{\xi\ge -1}]\le e^x\E[N(\xi)^{-\epsilon}]=\frac{e^{x}}{1-\epsilon}$,
we only need to focus on $\E[N(\xi)^{-\epsilon} e^{x\xi}\id_{\xi<- 1}]$.
By the fact that for any $\gamma>0$,
$\lim_{y\rightarrow -\infty}N(y) e^{\gamma y}
=\lim_{y\rightarrow -\infty}\frac{N'(y)}{-\gamma e^{-\gamma y}}=0$,
we deduce that  there exists a constant $K>0$ such that
$N(y)^{(1-\epsilon)/2} e^{xy}\le K$ $\forall y<-1$. Thus
$$\E[N(\xi)^{-\epsilon} e^{x\xi}\id_{\xi<- 1}]\le K\E[N(\xi)^{-(1+\epsilon)/2}]<+\infty.$$
The proof is complete.
\end{proof}

\begin{proposition}
If  $\inf_{t\in [0,T]} \delta(t)>0$, then $\tilde w$
satisfies Assumption  \ref{h13}-(iii).
\end{proposition}
\begin{proof} Take $\epsilon:=1\wedge \inf_{t\in [0,T]} \delta(t)>0$.
For any $t\in [0,T]$, it follows from the bound of $w'_{TK}(p;\delta)$ (see the proof of Proposition \ref{uandwiv}) that
\begin{eqnarray*}
h(t,2)&=&\E[w'_{TK}(N(\xi);\delta(t))e^{2\xi}]\\
&\le&\E[N(\xi)^{\delta(t)-1}e^{2\xi}]+\E[N(-\xi)^{\delta(t)-1}e^{2\xi}]\\
&=&\E[N(\xi)^{\delta(t)-1}e^{2\xi}]+\E[N(\xi)^{\delta(t)-1}e^{-2\xi}]\\
&\le &\E[N(\xi)^{\epsilon-1}e^{2\xi}]+\E[N(\xi)^{\epsilon-1}e^{-2\xi}].
\end{eqnarray*}
According to Lemma \ref{suff-(iii)}, we know $h(t,2)\le K$ for some constant $K$ independent of $t$. However, $h(t, x)$ is increasing in $x$; so $\sup_{t\in [0,T]} h(t,1)\le \sup_{t\in [0,T]}h(t, 2)<K$.

Next, by the fact that $\xi^2<e^\xi+e^{-\xi}$ we have
 $h^\dpm_x(t, 1)\le h(t,2)+h(t,0)<K+1$.
\end{proof}

Finally, we check Assumption \ref{h13}-(ii).
We first need two lemmas.

\begin{lemma}\label{nyd}
There exists $\bar\delta\in(0,1)$ such that
$$\sup_{\delta\in[\bar\delta,1]}\int_0^{+\infty} \ln[N(y)^{\delta}+N(-y)^{\delta}]dy<+\infty.$$
\end{lemma}
\begin{proof}
Fix $\delta\in(0,1].$ Observe that
\begin{eqnarray*}
\lim_{y\rightarrow +\infty}\frac{\ln[N(y)^{\delta}+N(-y)^{\delta}]}{y^{-2}}
&=&\lim_{y\rightarrow +\infty}\frac{\delta N'(y)}{y^{-3}}\frac{1}{-2}\frac{N(y)^{\delta-1}-N(-y)^{\delta-1}}{N(y)^{\delta}+N(-y)^{\delta}}\\
&\le&\lim_{y\rightarrow +\infty}\frac{\delta N'(y)}{y^{-3}}\frac{1}{2}\frac{N(-y)^{\delta-1}}{N(y)^{\delta}+N(-y)^{\delta}}\\
&\le&\lim_{y\rightarrow +\infty}\frac{N'(y)}{N(-y)^{1-\delta}y^{-3}}\frac{\delta }{2}\\
&=&\lim_{y\rightarrow +\infty}\frac{\delta }{2} \frac{N'(y)^{1-\delta/2}}{N(-y)^{1-\delta}}\frac{N'(y)^{\delta/2}}{y^{-3}}\\
&=&0.
\end{eqnarray*}
Hence, there exists a constant $K>0$ (which may depend on $\delta$), such that
$$\ln[N(y)^{\delta}+N(-y)^{\delta}]<K y^{-2} \quad \forall \, y>1.$$
As a result,
$$\int_0^{+\infty} \ln[N(y)^{\delta}+N(-y)^{\delta}]dy\le \int_0^1\ln(N(y)^{\delta}+N(-y)^{\delta})dy
+K\int_1^Ty^{-2}dy<+\infty.$$
Note that the integrand on the left hand side is decreasing in $\delta$; hence
the above finiteness is uniform when $\delta$ is sufficiently close to 1.
\end{proof}

Recall we have defined $h(x;\delta)=\E[w'_{TK}(N(\xi);\delta)e^{x\xi}]$, $x\in\R$, $\delta\in(0,1]$.
\begin{lemma}\label{hxdelta}
We have
$$ h'(0;\delta)=\frac{1-\delta}{\delta}\int_0^{+\infty} \ln[N(y)^{\delta}+N(-y)^{\delta}]dy+o\left(\frac{1-\delta}{\delta}\right)$$
when $\delta$ is sufficiently close to 1.
\end{lemma}

\begin{proof}
By (\ref{hderiv}) and (\ref{wkt}), we have
\begin{eqnarray*}
h'(0;\delta)&=&\int_0^{+\infty}(\delta-1)\alpha(N(y);\delta) [N(y)^{\delta-1}-N(-y)^{\delta-1}]ydN(y)\\
&=&\frac{\delta-1}{\delta} \int_0^{+\infty}\alpha(N(y);\delta) yd[N(y)^{\delta}+N(-y)^{\delta}]\\
&=&\frac{\delta-1}{\delta}\frac{\delta}{\delta-1}\int_0^{+\infty}y
d[N(y)^{\delta}+N(-y)^{\delta}]^{1-1/\delta}\\
&=&\int_0^{+\infty}\{1-[N(y)^{\delta}+N(-y)^{\delta}]^{1-1/\delta}\}dy,
\end{eqnarray*}
 where we have used the fact that $\lim_{y\rightarrow +\infty}\{1-[N(y)^{\delta}+N(-y)^{\delta}]^{1-1/\delta}\}y=0$.

 Applying the general Taylor expansion $x^\epsilon=1+\epsilon \ln x+o(|\epsilon|) (\ln x)^2$ for $x>0$ and sufficiently small $|\epsilon|$,  we deduce
 \[
 \int_0^{+\infty}\{1-[N(y)^{\delta}+N(-y)^{\delta}]^{1-1/\delta}\}dy=\frac{1-\delta}{\delta}\int_0^{+\infty} \ln[N(y)^{\delta}+N(-y)^{\delta}]dy+o\left(\frac{1-\delta}{\delta}\right),
 \]
where we have used the finiteness $\int_0^{+\infty}\{ \ln[N(y)^{\delta}+N(-y)^{\delta}]\}^2dy<+\infty$, which follows from
the inequalities $0<\ln[N(y)^{\delta}+N(-y)^{\delta}]<\ln 2$ and Lemma \ref{nyd}.
\end{proof}

\begin{proposition}
If $$\limsup_{t\uparrow T}\frac{1-\delta(t)}{\delta(t)\sqrt{T-t}}=0 \mbox{ and } \liminf_{t\uparrow T}|\theta(t)|^2>0, $$
then $\tilde w$ satisfies Assumption \ref{h13}-(ii).
\end{proposition}

\begin{proof} This follows immediately from Lemmas \ref{nyd} and \ref{hxdelta}.
\end{proof}

The second condition in the above is satisfied if we assume that $\theta(\cdot)$ is left continuous at $T$ and $\theta(T)\neq0$, which means that we do not have a trivial market at $T$. The first condition means that as $t$ approaches the terminal time $T$, $\delta(t)$ should approach $\delta(T)\equiv 1$ faster than $\sqrt{T-t}$.

To summarize, the family of time-varying TK weighting functions satisfy all the assumptions of the paper if the measurable function $\delta(\cdot)$ satisfies
\[ \inf_{t\in [0,T]} \delta(t)>0,\;\;\limsup_{t\uparrow T}\frac{1-\delta(t)}{\delta(t)\sqrt{T-t}}=0.
\]

\end{document}